\documentclass[prl,aps,twocolumn,twoside,superscriptaddress]{revtex4-1}

\usepackage{amssymb}
\usepackage{graphicx}
\usepackage{graphics}
\usepackage{amsmath}
\usepackage{amsthm}
\usepackage{color}
\usepackage{dsfont}

\usepackage{mathtools}

\def\>{\rangle}
\def\<{\langle}

\def\id{\mathsf{id}}
\def\mB{\mathcal{B}}
\def\mE{\mathcal{E}}

\def\mF{\mathcal{F}}
\def\mN{\mathcal{N}}

\def\mP{\mathcal{P}}
\def\mS{\mathcal{S}}
\def\mD{\mathcal{D}}

\renewcommand{\qedsymbol}{\nobreak \ifvmode \relax \else
	\ifdim \lastskip<1.5em \hskip-\lastskip \hskip1.5em plus0em
	minus0.5em \fi \nobreak \vrule height0.75em width0.5em
	depth0.25em\fi}

\renewcommand{\geq}{\geqslant}
\renewcommand{\leq}{\leqslant}

\newtheorem{theorem}{Theorem}
\newtheorem*{theorem*}{Theorem}

\newtheorem{lemma}[theorem]{Lemma}
\newtheorem*{lemma*}{Lemma}

\newtheorem{definition}[theorem]{Definition}
\newtheorem*{definition*}{Definition}

\theoremstyle{remark}
\newtheorem*{remark}{Remark}

\newcommand{\bea}{\begin{eqnarray}}
\newcommand{\eea}{\end{eqnarray}}
\newcommand{\be}{\begin{equation}}
\newcommand{\ee}{\end{equation}}
\newcommand{\ba}{\begin{equation}\begin{aligned}}
\newcommand{\ea}{\end{aligned}\end{equation}}

\def\be{\begin{equation}}
\def\ee{\end{equation}}

\newcommand{\cptp}{{\rm CPTP}}
\newcommand{\post}{{\rm post}}
\newcommand{\pree}{{\rm pre}}

\newcommand{\mA}{\mathcal{A}}

\newcommand{\mH}{\mathcal{H}}

\newcommand{\mM}{\mathcal{M}}

\newcommand{\lr}{\rangle\langle}

\newcommand{\ra}{\rangle}
\newcommand{\tr}{{\rm Tr}}

\newcommand{\mbb}[1]{\mathbb{#1}}

\newcommand{\ket}[1]{|#1\rangle}

\newcommand{\eqdef}{\coloneqq}

\def\tA{\tilde{A}}

\def\tR{\tilde{R}}
\def\trho{\tilde{\rho}}

\def\mf{\mathfrak{F}}

\def\cptp{{\rm CPTP}}

\newcommand{\ox}{\otimes}

\newcommand{\proj}[1]{| #1\rangle\!\langle #1 |}

\begin{document}
	
\title{How to quantify a dynamical quantum resource}

\author{Gilad Gour}
\email{giladgour@gmail.com}
\affiliation{Department of Mathematics and Statistics, Institute for Quantum Science and Technology, 
  University of Calgary, Calgary, Alberta T2N 1N4, Canada}

\author{Andreas Winter}
\email{andreas.winter@uab.cat}
\affiliation{ICREA \& F\'{i}sica Te\`{o}rica: Informaci\'{o} i Fen\`{o}mens Qu\`{a}ntics, 
  Departament de F\'{i}sica, Universitat Aut\`{o}noma de Barcelona, ES-08193 Bellaterra (Barcelona), Spain}
	
\date{7 June 2019}

\begin{abstract}
      We show that the generalization of the relative entropy of a resource 
      from states to channels is not unique, and there are at least six such 
      generalizations. We then show that two of these generalizations are 
      asymptotically continuous, satisfy a version of the asymptotic equipartition 
      property, and their regularizations appear in the power exponent of 
      channel versions of the quantum Stein's Lemma. To obtain our results, 
      we use a new type of ``smoothing" that can be applied to functions of 
      channels (with no state analog). We call it ``liberal smoothing" as 
      it allows for more spread in the optimization. Along the way, we show 
      that the diamond norm can be expressed as a max relative entropy distance to 
      the set of quantum channels, and prove a variety of properties of all 
      six generalizations of the relative entropy of a resource.
\end{abstract}
\maketitle

{\it Introduction--} 
In recent years it has been recognized that many properties of physical systems, 
such as quantum entanglement, asymmetry, coherence, athermality, contextuality, 
and many others, can be viewed as resources circumventing certain constraints 
imposed on physical systems (see~\cite{CG2019} and references therein). 
Each resource can be classified as being classical or quantum, static (e.g. 
entangled state) or dynamic (e.g. quantum channel), noisy or noiseless, 
leading to numerous interesting quantum information processing tasks~\cite{DHW2008} 
(e.g. quantum teleportation~\cite{BB1993}).
While there are many ways to quantify the resourcefulness of such properties, 
all quantifiers of a resource must satisfy certain conditions such as 
monotonicity under the set of free operations. Typically, there are  numerous 
measures that satisfy these conditions, but what can single out a given 
measure is an operational interpretation, giving it meaning beyond its 
sheer ability to quantify somewhat vaguely the resource.

The relative entropy of a resource, which was originally defined in~\cite{Vedral-1997a} 
for entanglement theory, is an example of a measure that has such an operational 
interpretation in many quantum resource theories (QRTs). First, it was shown 
in~\cite{HOHH2002,Gour-2009a} to be a unique measure in reversible QRTs, and 
then was shown to be the unique asymptotic rate of interconversion among 
static resources under resource non-generating operations~\cite{Brandao-2015a}. 
Moreover, it was shown very recently~\cite{BM2018,AHJ2018} that resource erasure as a universal 
operational task leads to the (regularized) relative entropy of a
resource as the optimal rate (this idea was first laid out in~\cite{GPW2005}).
In addition, this measure satisfies the asymptotic equipartition property 
(AEP)~\cite{BP2010}, appears as an optimal rate in the generalized quantum 
Stein's Lemma~\cite{BP2010}, and is asymptotically continuous~\cite{SH2006,C2006}, 
a property linked to it also being a non-lockable measure~\cite{HHHO2005}. 
Due to all of these properties, the relative entropy of a resource plays 
a major role in many QRTs~\cite{CG2019}.

In this paper we study six generalizations of the quantum relative entropy 
of a resource from static resources (i.e. states) to dynamic ones (i.e channels). 
Four of these measures were introduced very recently in~\cite{LY2019,LW2019}. 
We show that for two of them, the relative entropy of the dynamical resource 
is asymptotically continuous, satisfies a version of the AEP, and a version 
of their regularization appear as optimal rates in a version of the quantum 
Stein's Lemma for channels. In addition, we show that all these measures are 
indeed generalizations to dynamical resources in the sense that they reduce 
to the relative entropy of a static resource for replacement (i.e. constant) channels.

{\it Resource theories of quantum processes--}~\cite{LY2019,LW2019,GS2019,G2019,TEZ+2019}
A quantum resource theory (QRT), consists of a function $\mf$ taking any pair 
of physical systems $A$ and $B$ to a subset of completely positive and trace 
preserving (CPTP) maps $\mf(A\to B)\subset \cptp(A\to B)$, where $\cptp(A\to B)$ 
is the set of all CPTP maps (i.e. quantum channels) from $\mB(A)$ (bounded 
operators on Hilbert space of system $A$) to $\mB(B)$. 
The mapping $\mf$ is a quantum resource theory if the following two conditions hold:
\begin{enumerate}
  \item For any physical system $A$ the set $\mf(A\to A)$ contains the identity map $\id_A$. 
  \item For any three systems $A,B,C$, if $\mM\in\mf(A\to B)$ and $\mN\in\mf(B\to C)$ 
        then $\mN\circ\mM\in \mf(A\to C)$.
\end{enumerate}
Denoting by $1$ the trivial Hilbert space we identify $\mf(1\to A)$ with the set 
of free density matrices in $\mB(A)$. That is, a density matrix $\rho\in\mf(1\to A)$ 
can be viewed as the CPTP map $\rho(z)=z\rho$ for all $z\in\mbb{C}$. 
For simplicity, we will write $\mf(1\to A)\equiv\mf(A)$.
Typically, QRTs are physical in the sense that they arise from some physical 
constraints, and therefore admit a tensor product structure. That is, the set of 
free operations $\mf$ satisfies the following additional conditions:
\begin{enumerate}
  \item[3.] The free operations are ``completely free'': For any three physical 
          systems $A$, $B$, and $C$, if $\mM\in\mf(A\to B)$ then 
          $\id_C\otimes\mM\in\mf(CA\to CB)$.
  \item[4.] Discarding a system (i.e. the trace) is a free operation: For any 
          system $A$, the set $\mf(A\to 1)$ is not empty.
\end{enumerate}
The above additional conditions are very natural and satisfied by almost all QRTs 
studied in literature. They implies the following properties~\cite{CG2019}:
\begin{itemize}
  \item If $\mM_1$ and $\mM_2$ are free channels then also $\mM_1\otimes\mM_2$ is free. 
  \item Appending free states is a free operation: For any given free state 
        $\sigma\in\mf(B)$, the CPTP map $\mM_{\sigma}(\rho):=\rho\otimes\sigma$ 
        is a free map, i.e., it belongs to $\mf(A\to AB)$.
  \item The replacement map $\mM_\sigma(\rho)\eqdef\sigma$, for any density 
        matrix $\rho\in\mB(A)$ and a fixed free state $\sigma\in\mf(B)$, is a 
        free channel; i.e. $\mM_\sigma\in\mf(A\to B)$.
\end{itemize}
It is also physical to assume that $\mf(A\to B)$ is a closed set, since otherwise 
there exists a sequence of free channels whose limit is a resource channel. 
Finally, we will assume that for any integer $n$, free channel 
$\mN\in\mf(A_1\cdots A_n\to B_1\cdots B_n)$, and two permutation channels 
$\mP_A^{\pi}$ and $\mP_B^{\pi^{-1}}$ corresponding to a permutation $\pi$ 
on $n$ elements, we have
\be
  \mP_B^{\pi^{-1}}\circ\mN_{A_1\cdots A_n\to B_1\cdots B_n}\circ\mP_A^{\pi}
                               \in\mf(A_1\cdots A_n\to B_1\cdots B_n)\;.\nonumber
\ee
Note that almost all QRTs discussed in literature satisfy this last condition 
including entanglement theory, coherence, athermality, etc.
In the rest of this  paper we will assume that $\mf$ satisfies all the 
above conditions.

The most general physical operation that can be performed on a dynamical resource 
$\mN\in\cptp(A\to B)$ can be characterized with a superchannel~\cite{Pavia1,G2019}, 
$\Theta$, defined for all $\mN\in\cptp(A\to B)$ as a transformation of the form
\be
  \label{super}
  \Theta[\mN_{A\to B}]=\mE^{\post}_{BE\to B'}\circ\mN_{A\to B}\circ\mE^{\pree}_{A'\to AE}\;,
\ee
where $\mE^{\post}\in\cptp(BE\to B')$ and $\mE^{\pree}\in\cptp(A'\to AE)$ are 
quantum channels. We say that the superchannel $\Theta$ is free if in addition 
$\mE^{\post}\in\mf(BE\to B')$ and $\mE^{\pree}\in\mf(A'\to AE)$ 
(i.e. $\mE^{\post}$ and $\mE^{\pree}$ are free). 
Therefore, any measure of a resource $E:\cptp\to\mbb{R}$ must satisfy
\be
  \label{mono}
  E\big(\Theta[\mN_{A\to B}]\big)\leq E\big(\mN_{A\to B}\big)\;,
\ee
for all $\mN\in\cptp(A\to B)$ and all free superchannels $\Theta$. In addition, 
we require that $E(\mN)=0$ if $\mN\in\mf(A\to B)$. This condition implies that $E$ 
is non-negative. To see it, take $\mE^{\post}_{BE\to B'}$ in~\eqref{super} to 
be the replacement map whose output is some free state in $\mf(B')$, and observe 
that for this case $0=E(\Theta[\mN])\leq E(\mN)$ for all $\mN\in\cptp(A\to B)$.

{\it The relative entropy of a resource--}
We will consider here two generalization of the relative entropy of a resource 
from the state domain to the channel domain, and leave four further generalizations 
to the supplemental material (SM).
The first relative entropy of a dynamical resource $\mN\in\cptp(A\to B)$ is defined as
\begin{align}
  \label{maina}
  D_\mf(\mN)\eqdef\inf_{\mM\in\mf(A\to B)}D(\mN\|\mM)\;,
\end{align}
with the channel divergence~\cite{Cooney2016,Wilde2018,G2019}
\be
  D(\mN\|\mM) \eqdef \max_{\varphi\in\mD({RA})} 
                     D\left(\mN_{A\to B}(\varphi_{RA})\|\mM_{A\to B}(\varphi_{RA})\right)\;,
\ee
and $D(\rho\|\sigma)=\tr[\rho\log\rho-\rho\log\sigma]$ is the relative entropy.
The optimization is over all states $\varphi_{RA}$, where w.l.o.g. we can take 
$R\cong A$ and $\varphi_{RA}$ is pure~\cite{Cooney2016,Wilde2018}.
If the optimization over $\mD(RA)$ is replaced with optimization over the set 
of all density matrices $\mf(RA)$, then one gets the second generalization~\cite{LY2019}
\begin{align}
  \label{maind}
  E_\mf(\mN) \eqdef &\min_{\mM\in\mf(A\to B)}\sup_{\rho\in\mf(RA)}\nonumber\\
                    &D\left(\mN_{A\to B}(\rho_{RA})\|\mM_{A\to B}(\rho_{RA})\right)\;,
\end{align}
where the supremum is over all free states $\rho\in\mf(RA)$ and all dimensions 
$|R|$, and the minimum is over all free channels in $\mf(A\to B)$. Both $D_\mf$ 
and $E_\mf$, as well as other generalizations, were introduced very recently 
in~\cite{LY2019,LW2019}, and in the SM we list all of them along with a few 
new ones and discuss some of their properties. For clarity, we leave the 
technical details of all proofs to the SM.

\begin{theorem}
\label{properties}
The above relative entropies have the following properties:
\begin{enumerate}
\item {\rm\bf [Monotonicity]} $D_{\mf}$ and $E_\mf$ behave monotonically under 
      free superchannels. Specifically, let $\mE^{\post}\in\cptp(BE\to B')$ and 
      $\mE^{\pree}\in\cptp(A'\to AE)$ be completely resource RNG channels, and 
      let $\Theta$ has the form~\eqref{super}. Then, for all $\mN\in\cptp(A\to B)$
      \be
        D_\mf\big(\Theta[\mN]\big) \leq D_\mf\big(\mN\big)\;\;;\;\;E_\mf\big(\Theta[\mN]\big)
                                   \leq E_\mf\big(\mN\big)\;.
      \ee
\item {\rm\bf [Reduction]} Let $\mN\in\cptp(A\to B)$ be a constant channel 
      $\mN(X_A)=\tr[X_A]\omega_B$ for all $X_A\in\mB(A)$ and a fixed density 
      matrix $\omega_B\in\mD(B)$. Then,
      \be
        D_{\mf}(\mN)=E_{\mf}(\mN)=D_\mf(\omega_B)\eqdef\min_{\sigma\in\mf(B)}D(\omega_B\|\sigma_B)\;.
      \ee
\item {\rm\bf [Faithfulness]} $D_\mf(\mN_{A\to B})=0$ if and only if 
      $\mN\in\mf(A\to B)$.
      If $E_{\mf}(\mN)=0$ for some $\mN\in\cptp(A\to B)$ then $\mN$ must be 
      completely resource non-generating (RNG). Moreover, if for $|R|=|A|$ the set 
      $\mf(RA)$ contains a pure state with full Schmidt rank, then 
      \be
        E_\mf(\mN_{A\to B})=0\iff\mN\in\mf(A\to B)\;.
      \ee
\end{enumerate}
\end{theorem}

In contrast to the monotonicity property above, the function $D_\mf$ behaves 
monotonically under \emph{any} RNG superchannel. This follows directly from the following:
\ba\nonumber
  D_\mf(\Theta[\mN])&=    \min_{\Omega\in\mf(A'\to B')}D(\Theta[\mN_{A\to B}]\|\Omega_{A'\to B'})\\
                    &\leq \min_{\mM\in\mf(A\to B)}D(\Theta[\mN_{A\to B}]\|\Theta[\mM_{A\to B}])\\
                    &\leq \min_{\mM\in\mf(A\to B)}D(\mN_{A\to B}\|\mM_{A\to B})=D_{\mf}(\mN)\;,
\ea
where the first inequality follows from the fact that $\Theta$ is RNG, and the second 
from the data processing inequality of the channel divergence~\cite{G2019}. Note also 
that from their definitions we always have
\be
  E_{\mf}(\mN)\leq D_{\mf}(\mN)\quad\forall\;\mN\in\cptp(A\to B)\;.
\ee

One may wonder if exchanging the min-max order in~\eqref{maina} and~\eqref{maind} 
would yield other relative entropy based measures that are in general different 
than $D_\mf$ and $E_\mf$. However, in the following theorem we show that this is 
not the case.

\begin{theorem}
  \label{minmax}
  Let $d:\mD(A)\times\mD(A)\to\mbb{R}$ be any function satisfying non-negativity, 
  contractivity (monotonicity) under CPTP maps, and joint concavity under orthogonally 
  flagged mixtures:
  This means that for any two families $\{\rho_x\}$ and $\{\sigma_x\}$ of states, and
  any probability distribution $\{p_x\}$,
  \ba
    \label{eq:concave}
    d\Big( \sum_x p_x \rho_x\ox\proj{x}, \sum_x p_x \sigma_x&\ox\proj{x} \Big) \\
                                                            &\geq \sum_x p_x d(\rho_x,\sigma_x)\;,
  \ea
  where $|x\ra$ are orthonormal basis states of an auxiliary system. Moreover, 
  suppose $d$ is convex in the second argument, and suppose $\mf(A\to B)$ is convex. Then,
  \ba
    &\inf_{\mM\in\mf(A\to B)}\sup_{\rho\in\mf(RA)}
           d\left(\mN_{A\to B}(\rho_{RA}),\mM_{A\to B}(\rho_{RA})\right)\\
    &=\sup_{\rho\in\mf(RA)}\inf_{\mM\in\mf(A\to B)}
           d\left(\mN_{A\to B}(\rho_{RA}),\mM_{A\to B}(\rho_{RA})\right).\nonumber
  \ea
\end{theorem}
Note that the relative entropy $D$ (as well as the trace distance and all the 
Renyi divergences) satisfies~\eqref{eq:concave} with equality, and therefore 
$E_\mf$ and $D_\mf$ will not change by swapping the min-max order.

{\it Asymptotic continuity--}
Since we only consider here QRTs that admits the tensor product structure, the 
replacement channels $\mM_{\sigma}(X)=\tr[X]\sigma$ are free (i.e. in $\mf(A\to B)$) 
for any free $\sigma\in\mf(B)$. In the SM we show that this implies that $E_\mf$ 
is bounded as long as the set of free states contains a full rank state. For 
example, if $\mf(B)$ contains the maximally mixed (uniform) state $I_B/|B|$ 
(were $|B|$ is the dimension of system $B$), then
\be
  E_\mf(\mN)\leq D_\mf(\mN)\leq \log\big(|B|^2|A|\big)\;.
\ee
The fact that $E_\mf$ and $D_\mf$ are bounded enable us to prove that they 
are also asymptotically continuous.

\begin{definition}
  A function $E: \cptp\to \mbb{R}_+$ is said to be asymptotically continuous 
  if for any $\mM,\mN\in\cptp(A\to B)$,
  \be
    \left|E(\mM)-E(\mN)\right|\leq\log(|AB|)f\left(\left\|\mM-\mN\right\|_{\diamond}\right)\;,
  \ee
  where $f:\mbb{R}\to\mbb{R}$ is some function independent on the dimensions 
  and satisfies $\lim_{\epsilon\to 0^+}f(\epsilon)= 0$.
\end{definition}

\begin{theorem}
  \label{acn}
  Suppose that for any system $A$, $\mf(A)$ contains a full rank state. Then, 
  $D_\mf$ is asymptotically continuous. Moreover, if in addition, for any system 
  $A$ the extreme points of $\mf(A)$ are pure states (e.g. entanglement theory, 
  coherence, etc), then $E_\mf$ is also asymptotically continuous.
\end{theorem}

\begin{remark}
The proof of the theorem above is based on a key observation that the diamond 
norm can be expressed in terms of the max relative entropy distance of $\mN-\mM$ to the 
set of all quantum channels $Q(A\to B)$ (see SM for more details).
For $E_\mf$ the condition that the extreme points of the set of free states 
are pure states, ensures that the supremum in~\eqref{maind} can be replaced
with a maximum since in this case $|R|$ can be shown to be bounded by $|A|$. 
If the extreme points of the set of free states are not pure states but $|R|$ 
is polynomially bounded in $|AB|$, then also in this case $E_\mf$ is 
asymptotically continuous. This happens for example in the QRT of thermodynamics. Finally, we point out that asymptotic continuity for certain amortized measures of entanglement was recently proved in~\cite{KW2017}.
\end{remark}

{\it Asymptotic Equipartition Property (AEP)--}
The logarithmic robustness of a dynamical resource $\mN\in\cptp(A\to B)$ is 
defined as~\cite{LW2019}
\begin{align}
  &LR_\mf(\mN_{A\to B})\eqdef\min_{\mM\in\mf(A\to B)}D_{\max}(\mN_{A\to B}\|\mM_{A\to B})\nonumber\\
  &\eqdef\log_2\min\big\{t\;:\;t\mM\geq\mN\;\;;\;\; \mM\in\mf(A\to B)\big\}\;,
\end{align}
where the ordering $t\mM\geq\mN$ means that $t\mM-\mN$ is completely positive (CP). 
We also define here
\begin{align}
  &\underline{LR}_{\mf}(\mN_{A\to B})\eqdef\\
  &\min_{\mM\in\mf(A\to B)}\sup_{\varphi\in\mf(RA)}
       D_{\max}\big(\mN_{A\to B}(\varphi_{RA})\|\mM_{A\to B}(\varphi_{RA})\big)\;.\nonumber
\end{align}
Like $D_{\mf}$ and $E_{\mf}$, the functions  $LR_{\mf}$ and $\underline{LR}_{\mf}$ 
are resource monotones (see SM). Note that by Theorem~\ref{minmax} the order 
sup-min can be exchanged, and furthermore,
\be
  \underline{LR}_{\mf}\leq LR_{\mf}\;,
\ee
with equality if $\mf(RA)$ contains a pure state of full Schmidt rank. For example, 
in entanglement theory, system $A$ is replaced with $AB$ and $R$ with $R_AR_B$ 
so that $\mf(R_AR_BAB)$ contains the state 
$\phi^{+}_{(R_AR_B)(AB)}=\phi^{+}_{R_AA}\otimes\phi^{+}_{R_BB}$, where $\phi^+$ 
stands for the maximally entangled state between the respective spaces. Hence, 
$\phi^{+}_{(R_AR_B)(AB)}$ has full Schmidt rank between $R_AR_B$ and $AB$ 
(even though it is a product state between Alice ($R_AA$) and Bob ($R_BB$)). 
Therefore, in entanglement theory $\underline{LR}_{\mf}=LR_{\mf}$.  

The smoothed version of the logarithmic robustness can be defined as~\cite{LW2019}
\be
  \widetilde{LR}_\mf^{\epsilon}(\mN)\eqdef\min_{\mN'\in B_\epsilon(\mN)}LR_\mf(\mN')\;,
\ee
where
\be 
  B^\epsilon(\mN)\eqdef\Big\{\mN'\in\cptp(A\to B)\;:\;\|\mN'-\mN\|_{\diamond}\leq \epsilon\Big\}.
\ee
The above diamond-smoothed log-robustness is a straightforward generalization 
from states to channels, and has an operational interpretation in the setting 
of resource erasure~\cite{LW2019}, generalizing the single-shot part 
of~\cite{AHJ2018}.  However, our goal here is to define a method for smoothing 
that is the least restrictive possible. This will be necessary for a proof of 
an AEP for the logarithmic robustness of channels.

For this reason, we consider another (more ``liberal'') way to define smoothing 
for channels for which there is no analog in the state domain. For any state 
$\varphi\in\mD(RA)$ and a channel $\mN\in\cptp(A\to B)$ define 
$B_{\epsilon}^{\varphi}(\mN)$ to be the set of all CP maps (not necessarily 
trace preserving) $\mN'\in \text{CP}(A\to B)$ satisfying
\be
  \|\mN'_{A\to B}(\varphi_{RA})-\mN_{A\to B}(\varphi_{RA})\|_{1}\leq \epsilon\;.
\ee
Clearly,
$B_\epsilon(\mN)\subset\bigcap_{\varphi\in\mD(RA)}B_\epsilon^\varphi(\mN)$.
We define the smoothing of $LR_\mf$ as
\be
  \label{smooth}
  LR_{\mf}^{\epsilon}(\mN)
    \eqdef \max_{\varphi\in\mD(RA)}\min_{\mN'\in B_\epsilon^{\varphi}(\mN)}LR_\mf(\mN')\;.
\ee
Similarly, we denote by $\underline{LR}_{\mf}^{\epsilon}$ the above smoothing 
of $\underline{LR}_\mf$.
Note that the above types of smoothing respect the condition that for 
$\epsilon=0$ the smoothed quantities reduce to the non-smoothed ones.
Furthermore, from its definition it follows that (see SM for more details) 
\be
  LR_{\mf}^{\epsilon}(\mN)\leq \widetilde{LR}_{\mf}^{\epsilon}(\mN)\;,
\ee
justifying the name ``liberal smoothin''.

In the SM we show that $LR_\mf^{\epsilon}(\mN)$ is a resource monotone, 
and the regularized versions
\be
  LR_\mf^{\infty}(\mN) \eqdef\varliminf_{n\to\infty}
    \frac{LR_{\mf}^{\epsilon}(\mN^{\otimes n})}{n}\;;\;
  D_\mf^{\infty}(\mN)  \eqdef\varliminf_{n\to\infty}
    \frac{D_{\mf}^{\epsilon}(\mN^{\otimes n})}{n},\nonumber
\ee
satisfy $D_\mf^{\infty}(\mN)\leq LR_\mf^{\infty}(\mN)$. We believe that in 
general this inequality can be strict.
However, as we show now, if we revise also the type of regularization, then 
it is possible to get an equality.

The type of regularization that we consider here is as follows. For each 
$n\in\mbb{N}$, and a channel $\mN\in\cptp(A\to B)$, we define the quantities
\begin{align}
  D_{\mf}^{(n)}(\mN)\eqdef &\frac{1}{n}\max_{\varphi\in\mD(RA)}\min_{\mM\in\mf(A^n\to B^n)}\nonumber\\
                           &D\left(\mN_{A\to B}^{\otimes n}(\varphi_{RA}^{\otimes n})
                                     \big\|\mM_{A^n\to B^n}(\varphi_{RA}^{\otimes n})\right)\;,
\end{align}
and $E_{\mf}^{(n)}$ is defined exactly as above with $\mf(RA)$ replacing $\mD(RA)$.

In the SM we show that the limit $n\to\infty$ of $D_{\mf}^{(n)}$ and $E_{\mf}^{(n)}$ exists.
We therefore define the ``regularized" version of $D_{\mf}$ and $E_\mf$ to be
\be
  D_{\mf}^{(\infty)}(\mN)=\lim_{n\to\infty}D_{\mf}^{(n)}(\mN)\;\;;\;\;
  E_{\mf}^{(\infty)}(\mN)=\lim_{n\to\infty}E_{\mf}^{(n)}(\mN)\;.\nonumber
\ee
 
We can use this regularization method also for the liberal smoothed logarithmic 
robustness quantities $LR_{\mf}^{\epsilon}$ and $\underline{LR}_{\mf}^{\epsilon}$.
We define
\begin{align}
  &LR_{\mf}^{\epsilon,n}(\mN) 
      \eqdef \frac{1}{n}\max_{\varphi\in\mD(RA)}
                       \min_{\mN'\in B_\epsilon^{\varphi^{\otimes n}}(\mN^{\otimes n})}LR_{\mf}(\mN') \\
  &\text{and }
     LR^{(\infty)}_{\mf}
      \eqdef \lim_{\epsilon\to 0}\varliminf_{n\to\infty}{LR}_{\mf}^{\epsilon,n}(\mN)\;.
\end{align}
The quantities $\underline{LR}_{\mf}^{\epsilon,n}$ and $\underline{LR}_{\mf}^{(\infty)}$ 
are defined analogously with $\mf(RA)$ replacing $\mD(RA)$.
\begin{theorem}
  \label{aep}
  For all $\mN\in\cptp(A\to B)$
  \begin{align}
  D_{\mf}^{(\infty)}(\mN)
     = \lim_{\epsilon\to 0}\varlimsup_{n\to\infty}\frac{1}{n}{LR}_{\mf}^{\epsilon,n}(\mN^{\otimes n})
     = LR^{(\infty)}_{\mf}(\mN)\;.\nonumber
  \end{align}
  Moreover, if for any system $A$ the extreme points of $\mf(A)$ are pure states then
  \begin{align}
    E_{\mf}^{(\infty)}(\mN)
       =\lim_{\epsilon\to 0}\varlimsup_{n\to\infty}
                \frac{1}{n}\underline{LR}_{\mf}^{\epsilon,n}(\mN^{\otimes n})
       =\underline{LR}^{(\infty)}_{\mf}(\mN)\;.\nonumber
  \end{align}
\end{theorem}

{\it Quantum Channel Stein's Lemma--}
(See related work~\cite{Cooney2016,Wilde2018,Hayashi2009,DFY2009,G2019}.)
Consider the task of discriminating between $n$ copies of a fixed channel 
$\mN\in\cptp(A\to B)$ and one of the free channels in $\mf(A^n\to B^n)$.
There are two types of errors in such a task:
\begin{enumerate}
  \item The observer guesses that the channel belongs to $\mf(A^n\to B^n)$ 
        while the channel is $\mN^{\otimes n}_{A\to B}$. This occurs with probability
        \be\nonumber
          \alpha^{(n)}(\mN,P_n,\varphi_{RA})
            \eqdef\tr\left[\mN^{\otimes n}_{A\to B}\left(\varphi_{RA}^{\otimes n}\right)(I-P_n)\right]\;.
        \ee
        Here we consider the ``parallel'' case, in which the observer only 
        provides $n$ copies of a free state $\varphi\in\mf(RA)$, 
        and $0\leq P_n\leq I_{A^nB^n}$.
  \item The observer guesses that the channel is $\mN^{\otimes n}_{A\to B}$ 
        while the channel is some $\mM_n\in\mf(A^n\to B^n)$. This occurs with probability
        \be\nonumber
          \beta^{(n)}(P_n,\mM_n,\varphi_{RA})
            \eqdef\tr\left[\mM_{n}\left(\varphi_{RA}^{\otimes n}\right)P_n\right]\;,
        \ee
        and the worst case for a given $\varphi\in\mf(RA)$ is
        \be\nonumber
          \beta_{\mf}^{(n)}(P_n,\varphi_{RA})
          \eqdef\max_{\mM_n\in\mf(A^n\to B^n)}
                 \tr\left[\mM_{n}\left(\varphi_{RA}^{\otimes n}\right)P_n\right]\;.
        \ee
\end{enumerate}

We further define
\be
  \beta_{\mf,\epsilon}^{(n)}\left(\mN,\varphi_{RA}\right)\eqdef\min\beta_{\mf}^{(n)}(P_n,\varphi_{RA})\;,
\ee
where the minimum is over all $P_n$ satisfying $\alpha^{(n)}(\mN,P_n,\varphi_{RA})\leq\epsilon$ 
and $0\leq P_n\leq I_{R^nB^n}$.
\begin{theorem}
  \label{stein}
  Let $\mf$ be a convex resource theory satisfying all the conditions discussed 
  in the introduction, and suppose further that the set of free states contains 
  a full rank state. Then, for all $\epsilon\in(0,1)$,
  \be
    \label{ggg}
    \tilde{E}_\mf^{(\infty)}(\mN)
      = \max_{\varphi\in\mf(RA)}\lim_{n\to\infty}
         -\frac{\log\beta_{\mf,\epsilon}^{(n)}\left(\mN,\varphi_{RA}\right)}{n},
  \ee
  where
  \begin{align*}
    &\tilde{E}_\mf^{(\infty)}(\mN)\eqdef\nonumber\\
    &\max_{\varphi\in\mf(RA)}\lim_{n\to\infty}\min_{\mM\in\mf(A^n\to B^n)}
      \frac{D\left(\mN^{\otimes n}(\varphi_{RA}^{\otimes n})\big\|\mM(\varphi_{RA}^{\otimes n})\right)}{n}\;.
      \nonumber
\end{align*}
\end{theorem}
Note that the only difference between $\tilde{E}_\mf^{(\infty)}(\mN)$ and 
$E_\mf^{(\infty)}(\mN)$ is the order between the limit and the maximum. Therefore, 
we must have $\tilde{E}_\mf^{(\infty)}(\mN)\leq {E}_\mf^{(\infty)}(\mN)$, 
and it is left open to determine if this inequality can be strict. If the 
latter holds that would mean that $\tilde{E}_\mf^{(\infty)}(\mN)$ is yet 
another (distinct) generalization of the relative entropy of a resource.

{\it Conclusions--} 
We have seen that $D_\mf$ and $E_\mf$ are asymptotically continuous, satisfy 
the AEP, and are related to a channel-version of the quantum Stein's Lemma. 
To establish these results, we had to adopt two unconventional strategies, 
liberal smoothing and product-state channel regularization. In this way, 
lots of the properties in the state domain carry over to the channel domain.
 In the SM we also introduce additional four generalizations of the relative 
 entropy of a resource. This variety of generalizations indicates that in 
 the channel domain things are much more complicated. We believe that the 
 results and techniques presented here will provide an initial step towards 
 the development of QRT with dynamical resources.

\medskip
\acknowledgements
GG acknowledges support from the Natural Sciences and Engineering Research 
Council of Canada (NSERC).
AW was supported by the Spanish MINECO (project FIS2016-86681-P) with the 
support of FEDER funds, and the Generalitat de Catalunya (project 2017-SGR-1127).

\bibliography{QRTbib}

\begin{thebibliography}{34}%
\makeatletter
\providecommand \@ifxundefined [1]{%
 \@ifx{#1\undefined}
}%
\providecommand \@ifnum [1]{%
 \ifnum #1\expandafter \@firstoftwo
 \else \expandafter \@secondoftwo
 \fi
}%
\providecommand \@ifx [1]{%
 \ifx #1\expandafter \@firstoftwo
 \else \expandafter \@secondoftwo
 \fi
}%
\providecommand \natexlab [1]{#1}%
\providecommand \enquote  [1]{``#1''}%
\providecommand \bibnamefont  [1]{#1}%
\providecommand \bibfnamefont [1]{#1}%
\providecommand \citenamefont [1]{#1}%
\providecommand \href@noop [0]{\@secondoftwo}%
\providecommand \href [0]{\begingroup \@sanitize@url \@href}%
\providecommand \@href[1]{\@@startlink{#1}\@@href}%
\providecommand \@@href[1]{\endgroup#1\@@endlink}%
\providecommand \@sanitize@url [0]{\catcode `\\12\catcode `\$12\catcode
  `\&12\catcode `\#12\catcode `\^12\catcode `\_12\catcode `\%12\relax}%
\providecommand \@@startlink[1]{}%
\providecommand \@@endlink[0]{}%
\providecommand \url  [0]{\begingroup\@sanitize@url \@url }%
\providecommand \@url [1]{\endgroup\@href {#1}{\urlprefix }}%
\providecommand \urlprefix  [0]{URL }%
\providecommand \Eprint [0]{\href }%
\providecommand \doibase [0]{http://dx.doi.org/}%
\providecommand \selectlanguage [0]{\@gobble}%
\providecommand \bibinfo  [0]{\@secondoftwo}%
\providecommand \bibfield  [0]{\@secondoftwo}%
\providecommand \translation [1]{[#1]}%
\providecommand \BibitemOpen [0]{}%
\providecommand \bibitemStop [0]{}%
\providecommand \bibitemNoStop [0]{.\EOS\space}%
\providecommand \EOS [0]{\spacefactor3000\relax}%
\providecommand \BibitemShut  [1]{\csname bibitem#1\endcsname}%
\let\auto@bib@innerbib\@empty
\bibitem [{\citenamefont {Chitambar}\ and\ \citenamefont
  {Gour}(2019)}]{CG2019}%
  \BibitemOpen
  \bibfield  {author} {\bibinfo {author} {\bibfnamefont {E.}~\bibnamefont
  {Chitambar}}\ and\ \bibinfo {author} {\bibfnamefont {G.}~\bibnamefont
  {Gour}},\ }\href {\doibase 10.1103/RevModPhys.91.025001} {\bibfield
  {journal} {\bibinfo  {journal} {Rev. Mod. Phys.}\ }\textbf {\bibinfo {volume}
  {91}},\ \bibinfo {pages} {025001} (\bibinfo {year} {2019})}\BibitemShut
  {NoStop}%
\bibitem [{\citenamefont {Devetak}\ \emph {et~al.}(2008)\citenamefont
  {Devetak}, \citenamefont {Harrow},\ and\ \citenamefont {Winter}}]{DHW2008}%
  \BibitemOpen
  \bibfield  {author} {\bibinfo {author} {\bibfnamefont {I.}~\bibnamefont
  {Devetak}}, \bibinfo {author} {\bibfnamefont {A.~W.}\ \bibnamefont {Harrow}},
  \ and\ \bibinfo {author} {\bibfnamefont {A.~J.}\ \bibnamefont {Winter}},\
  }\href {\doibase 10.1109/TIT.2008.928980} {\bibfield  {journal} {\bibinfo
  {journal} {IEEE Transactions on Information Theory}\ }\textbf {\bibinfo
  {volume} {54}},\ \bibinfo {pages} {4587} (\bibinfo {year}
  {2008})}\BibitemShut {NoStop}%
\bibitem [{\citenamefont {Bennett}\ \emph {et~al.}(1993)\citenamefont
  {Bennett}, \citenamefont {Brassard}, \citenamefont {Cr\'epeau}, \citenamefont
  {Jozsa}, \citenamefont {Peres},\ and\ \citenamefont {Wootters}}]{BB1993}%
  \BibitemOpen
  \bibfield  {author} {\bibinfo {author} {\bibfnamefont {C.~H.}\ \bibnamefont
  {Bennett}}, \bibinfo {author} {\bibfnamefont {G.}~\bibnamefont {Brassard}},
  \bibinfo {author} {\bibfnamefont {C.}~\bibnamefont {Cr\'epeau}}, \bibinfo
  {author} {\bibfnamefont {R.}~\bibnamefont {Jozsa}}, \bibinfo {author}
  {\bibfnamefont {A.}~\bibnamefont {Peres}}, \ and\ \bibinfo {author}
  {\bibfnamefont {W.~K.}\ \bibnamefont {Wootters}},\ }\href {\doibase
  10.1103/PhysRevLett.70.1895} {\bibfield  {journal} {\bibinfo  {journal}
  {Phys. Rev. Lett.}\ }\textbf {\bibinfo {volume} {70}},\ \bibinfo {pages}
  {1895} (\bibinfo {year} {1993})}\BibitemShut {NoStop}%
\bibitem [{\citenamefont {Vedral}\ \emph {et~al.}(1997)\citenamefont {Vedral},
  \citenamefont {Plenio}, \citenamefont {Rippin},\ and\ \citenamefont
  {Knight}}]{Vedral-1997a}%
  \BibitemOpen
  \bibfield  {author} {\bibinfo {author} {\bibfnamefont {V.}~\bibnamefont
  {Vedral}}, \bibinfo {author} {\bibfnamefont {M.~B.}\ \bibnamefont {Plenio}},
  \bibinfo {author} {\bibfnamefont {M.~A.}\ \bibnamefont {Rippin}}, \ and\
  \bibinfo {author} {\bibfnamefont {P.~L.}\ \bibnamefont {Knight}},\ }\href
  {\doibase 10.1103/PhysRevLett.78.2275} {\bibfield  {journal} {\bibinfo
  {journal} {Phys. Rev. Lett.}\ }\textbf {\bibinfo {volume} {78}},\ \bibinfo
  {pages} {2275} (\bibinfo {year} {1997})}\BibitemShut {NoStop}%
\bibitem [{\citenamefont {Horodecki}\ \emph {et~al.}(2002)\citenamefont
  {Horodecki}, \citenamefont {Oppenheim},\ and\ \citenamefont
  {Horodecki}}]{HOHH2002}%
  \BibitemOpen
  \bibfield  {author} {\bibinfo {author} {\bibfnamefont {M.}~\bibnamefont
  {Horodecki}}, \bibinfo {author} {\bibfnamefont {J.}~\bibnamefont
  {Oppenheim}}, \ and\ \bibinfo {author} {\bibfnamefont {R.}~\bibnamefont
  {Horodecki}},\ }\href {\doibase 10.1103/PhysRevLett.89.240403} {\bibfield
  {journal} {\bibinfo  {journal} {Phys. Rev. Lett.}\ }\textbf {\bibinfo
  {volume} {89}},\ \bibinfo {pages} {240403} (\bibinfo {year}
  {2002})}\BibitemShut {NoStop}%
\bibitem [{\citenamefont {Gour}\ \emph {et~al.}(2009)\citenamefont {Gour},
  \citenamefont {Marvian},\ and\ \citenamefont {Spekkens}}]{Gour-2009a}%
  \BibitemOpen
  \bibfield  {author} {\bibinfo {author} {\bibfnamefont {G.}~\bibnamefont
  {Gour}}, \bibinfo {author} {\bibfnamefont {I.}~\bibnamefont {Marvian}}, \
  and\ \bibinfo {author} {\bibfnamefont {R.~W.}\ \bibnamefont {Spekkens}},\
  }\href {\doibase 10.1103/PhysRevA.80.012307} {\bibfield  {journal} {\bibinfo
  {journal} {Phys. Rev. A}\ }\textbf {\bibinfo {volume} {80}},\ \bibinfo
  {pages} {012307} (\bibinfo {year} {2009})}\BibitemShut {NoStop}%
\bibitem [{\citenamefont {Brand\~ao}\ and\ \citenamefont
  {Gour}(2015)}]{Brandao-2015a}%
  \BibitemOpen
  \bibfield  {author} {\bibinfo {author} {\bibfnamefont {F.~G. S.~L.}\
  \bibnamefont {Brand\~ao}}\ and\ \bibinfo {author} {\bibfnamefont
  {G.}~\bibnamefont {Gour}},\ }\href {\doibase 10.1103/PhysRevLett.115.070503}
  {\bibfield  {journal} {\bibinfo  {journal} {Phys. Rev. Lett.}\ }\textbf
  {\bibinfo {volume} {115}},\ \bibinfo {pages} {070503} (\bibinfo {year}
  {2015})}\BibitemShut {NoStop}%
\bibitem [{\citenamefont {Berta}\ and\ \citenamefont {Majenz}(2018)}]{BM2018}%
  \BibitemOpen
  \bibfield  {author} {\bibinfo {author} {\bibfnamefont {M.}~\bibnamefont
  {Berta}}\ and\ \bibinfo {author} {\bibfnamefont {C.}~\bibnamefont {Majenz}},\
  }\href {\doibase 10.1103/PhysRevLett.121.190503} {\bibfield  {journal}
  {\bibinfo  {journal} {Phys. Rev. Lett.}\ }\textbf {\bibinfo {volume} {121}},\
  \bibinfo {pages} {190503} (\bibinfo {year} {2018})}\BibitemShut {NoStop}%
\bibitem [{\citenamefont {Anshu}\ \emph {et~al.}(2018)\citenamefont {Anshu},
  \citenamefont {Hsieh},\ and\ \citenamefont {Jain}}]{AHJ2018}%
  \BibitemOpen
  \bibfield  {author} {\bibinfo {author} {\bibfnamefont {A.}~\bibnamefont
  {Anshu}}, \bibinfo {author} {\bibfnamefont {M.}~\bibnamefont {Hsieh}}, \ and\
  \bibinfo {author} {\bibfnamefont {R.}~\bibnamefont {Jain}},\ }\href {\doibase
  10.1103/PhysRevLett.121.190504} {\bibfield  {journal} {\bibinfo  {journal}
  {Phys. Rev. Lett.}\ }\textbf {\bibinfo {volume} {121}},\ \bibinfo {pages}
  {190504} (\bibinfo {year} {2018})}\BibitemShut {NoStop}%
\bibitem [{\citenamefont {Groisman}\ \emph {et~al.}(2005)\citenamefont
  {Groisman}, \citenamefont {Popescu},\ and\ \citenamefont {Winter}}]{GPW2005}%
  \BibitemOpen
  \bibfield  {author} {\bibinfo {author} {\bibfnamefont {B.}~\bibnamefont
  {Groisman}}, \bibinfo {author} {\bibfnamefont {S.}~\bibnamefont {Popescu}}, \
  and\ \bibinfo {author} {\bibfnamefont {A.}~\bibnamefont {Winter}},\ }\href
  {\doibase 10.1103/PhysRevA.72.032317} {\bibfield  {journal} {\bibinfo
  {journal} {Phys. Rev. A}\ }\textbf {\bibinfo {volume} {72}},\ \bibinfo
  {pages} {032317} (\bibinfo {year} {2005})}\BibitemShut {NoStop}%
\bibitem [{\citenamefont {Brand{\~a}o}\ and\ \citenamefont
  {Plenio}(2010)}]{BP2010}%
  \BibitemOpen
  \bibfield  {author} {\bibinfo {author} {\bibfnamefont {F.~G. S.~L.}\
  \bibnamefont {Brand{\~a}o}}\ and\ \bibinfo {author} {\bibfnamefont {M.~B.}\
  \bibnamefont {Plenio}},\ }\href {\doibase 10.1007/s00220-010-1005-z}
  {\bibfield  {journal} {\bibinfo  {journal} {Communications in Mathematical
  Physics}\ }\textbf {\bibinfo {volume} {295}},\ \bibinfo {pages} {791}
  (\bibinfo {year} {2010})}\BibitemShut {NoStop}%
\bibitem [{\citenamefont {Synak-Radtke}\ and\ \citenamefont
  {Horodecki}(2006)}]{SH2006}%
  \BibitemOpen
  \bibfield  {author} {\bibinfo {author} {\bibfnamefont {B.}~\bibnamefont
  {Synak-Radtke}}\ and\ \bibinfo {author} {\bibfnamefont {M.}~\bibnamefont
  {Horodecki}},\ }\href {\doibase 10.1088/0305-4470/39/26/l02} {\bibfield
  {journal} {\bibinfo  {journal} {Journal of Physics A: Mathematical and
  General}\ }\textbf {\bibinfo {volume} {39}},\ \bibinfo {pages} {L423}
  (\bibinfo {year} {2006})}\BibitemShut {NoStop}%
\bibitem [{\citenamefont {Christandl}(2006)}]{C2006}%
  \BibitemOpen
  \bibfield  {author} {\bibinfo {author} {\bibfnamefont {M.}~\bibnamefont
  {Christandl}},\ }\href@noop {} {\bibfield  {journal} {\bibinfo  {journal}
  {PhD thesis, Cambridge University, UK}\ } (\bibinfo {year} {2006})},\ \Eprint
  {http://arxiv.org/abs/arXiv:quant-ph/0604183} {arXiv:quant-ph/0604183}
  \BibitemShut {NoStop}%
\bibitem [{\citenamefont {Horodecki}\ \emph {et~al.}(2005)\citenamefont
  {Horodecki}, \citenamefont {Horodecki}, \citenamefont {Horodecki},\ and\
  \citenamefont {Oppenheim}}]{HHHO2005}%
  \BibitemOpen
  \bibfield  {author} {\bibinfo {author} {\bibfnamefont {K.}~\bibnamefont
  {Horodecki}}, \bibinfo {author} {\bibfnamefont {M.}~\bibnamefont
  {Horodecki}}, \bibinfo {author} {\bibfnamefont {P.}~\bibnamefont
  {Horodecki}}, \ and\ \bibinfo {author} {\bibfnamefont {J.}~\bibnamefont
  {Oppenheim}},\ }\href {\doibase 10.1103/PhysRevLett.94.200501} {\bibfield
  {journal} {\bibinfo  {journal} {Phys. Rev. Lett.}\ }\textbf {\bibinfo
  {volume} {94}},\ \bibinfo {pages} {200501} (\bibinfo {year}
  {2005})}\BibitemShut {NoStop}%
\bibitem [{\citenamefont {Liu}\ and\ \citenamefont {Yuan}(2019)}]{LY2019}%
  \BibitemOpen
  \bibfield  {author} {\bibinfo {author} {\bibfnamefont {Y.}~\bibnamefont
  {Liu}}\ and\ \bibinfo {author} {\bibfnamefont {X.}~\bibnamefont {Yuan}},\
  }\href@noop {} {\  (\bibinfo {year} {2019})},\ \Eprint
  {http://arxiv.org/abs/arXiv:1904.02680} {arXiv:1904.02680} \BibitemShut
  {NoStop}%
\bibitem [{\citenamefont {Liu}\ and\ \citenamefont {Winter}(2019)}]{LW2019}%
  \BibitemOpen
  \bibfield  {author} {\bibinfo {author} {\bibfnamefont {Z.-W.}\ \bibnamefont
  {Liu}}\ and\ \bibinfo {author} {\bibfnamefont {A.}~\bibnamefont {Winter}},\
  }\href@noop {} {\  (\bibinfo {year} {2019})},\ \Eprint
  {http://arxiv.org/abs/arXiv:1904.04201} {arXiv:1904.04201} \BibitemShut
  {NoStop}%
\bibitem [{\citenamefont {Gour}\ and\ \citenamefont {Scandolo}(2019)}]{GS2019}%
  \BibitemOpen
  \bibfield  {author} {\bibinfo {author} {\bibfnamefont {G.}~\bibnamefont
  {Gour}}\ and\ \bibinfo {author} {\bibfnamefont {C.~M.}\ \bibnamefont
  {Scandolo}},\ }\href@noop {} {\  (\bibinfo {year} {2019})},\ \Eprint
  {http://arxiv.org/abs/arXiv:1907.02552} {arXiv:1907.02552} \BibitemShut
  {NoStop}%
\bibitem [{\citenamefont {{Gour}}(2019)}]{G2019}%
  \BibitemOpen
  \bibfield  {author} {\bibinfo {author} {\bibfnamefont {G.}~\bibnamefont
  {{Gour}}},\ }\href {\doibase 10.1109/TIT.2019.2907989} {\bibfield  {journal}
  {\bibinfo  {journal} {IEEE Transactions on Information Theory}\ ,\ \bibinfo
  {pages} {1}} (\bibinfo {year} {2019})},\ \Eprint
  {http://arxiv.org/abs/arXiv:1808.02607} {arXiv:1808.02607} \BibitemShut
  {NoStop}%
\bibitem [{\citenamefont {Theurer}\ \emph {et~al.}(2019)\citenamefont
  {Theurer}, \citenamefont {Egloff}, \citenamefont {Zhang},\ and\ \citenamefont
  {Plenio}}]{TEZ+2019}%
  \BibitemOpen
  \bibfield  {author} {\bibinfo {author} {\bibfnamefont {T.}~\bibnamefont
  {Theurer}}, \bibinfo {author} {\bibfnamefont {D.}~\bibnamefont {Egloff}},
  \bibinfo {author} {\bibfnamefont {L.}~\bibnamefont {Zhang}}, \ and\ \bibinfo
  {author} {\bibfnamefont {M.~B.}\ \bibnamefont {Plenio}},\ }\href {\doibase
  10.1103/PhysRevLett.122.190405} {\bibfield  {journal} {\bibinfo  {journal}
  {Phys. Rev. Lett.}\ }\textbf {\bibinfo {volume} {122}},\ \bibinfo {pages}
  {190405} (\bibinfo {year} {2019})}\BibitemShut {NoStop}%
\bibitem [{\citenamefont {Chiribella}\ \emph {et~al.}(2008)\citenamefont
  {Chiribella}, \citenamefont {D'Ariano},\ and\ \citenamefont
  {Perinotti}}]{Pavia1}%
  \BibitemOpen
  \bibfield  {author} {\bibinfo {author} {\bibfnamefont {G.}~\bibnamefont
  {Chiribella}}, \bibinfo {author} {\bibfnamefont {G.~M.}\ \bibnamefont
  {D'Ariano}}, \ and\ \bibinfo {author} {\bibfnamefont {P.}~\bibnamefont
  {Perinotti}},\ }\href {http://stacks.iop.org/0295-5075/83/i=3/a=30004}
  {\bibfield  {journal} {\bibinfo  {journal} {EPL (Europhysics Letters)}\
  }\textbf {\bibinfo {volume} {83}},\ \bibinfo {pages} {30004} (\bibinfo {year}
  {2008})}\BibitemShut {NoStop}%
\bibitem [{\citenamefont {Cooney}\ \emph {et~al.}(2016)\citenamefont {Cooney},
  \citenamefont {Mosonyi},\ and\ \citenamefont {Wilde}}]{Cooney2016}%
  \BibitemOpen
  \bibfield  {author} {\bibinfo {author} {\bibfnamefont {T.}~\bibnamefont
  {Cooney}}, \bibinfo {author} {\bibfnamefont {M.}~\bibnamefont {Mosonyi}}, \
  and\ \bibinfo {author} {\bibfnamefont {M.~M.}\ \bibnamefont {Wilde}},\ }\href
  {\doibase 10.1007/s00220-016-2645-4} {\bibfield  {journal} {\bibinfo
  {journal} {Communications in Mathematical Physics}\ }\textbf {\bibinfo
  {volume} {344}},\ \bibinfo {pages} {797} (\bibinfo {year}
  {2016})}\BibitemShut {NoStop}%
\bibitem [{\citenamefont {Leditzky}\ \emph {et~al.}(2018)\citenamefont
  {Leditzky}, \citenamefont {Kaur}, \citenamefont {Datta},\ and\ \citenamefont
  {Wilde}}]{Wilde2018}%
  \BibitemOpen
  \bibfield  {author} {\bibinfo {author} {\bibfnamefont {F.}~\bibnamefont
  {Leditzky}}, \bibinfo {author} {\bibfnamefont {E.}~\bibnamefont {Kaur}},
  \bibinfo {author} {\bibfnamefont {N.}~\bibnamefont {Datta}}, \ and\ \bibinfo
  {author} {\bibfnamefont {M.~M.}\ \bibnamefont {Wilde}},\ }\href {\doibase
  10.1103/PhysRevA.97.012332} {\bibfield  {journal} {\bibinfo  {journal} {Phys.
  Rev. A}\ }\textbf {\bibinfo {volume} {97}},\ \bibinfo {pages} {012332}
  (\bibinfo {year} {2018})}\BibitemShut {NoStop}%
\bibitem [{\citenamefont {Kaur}\ and\ \citenamefont {Wilde}(2017)}]{KW2017}%
  \BibitemOpen
  \bibfield  {author} {\bibinfo {author} {\bibfnamefont {E.}~\bibnamefont
  {Kaur}}\ and\ \bibinfo {author} {\bibfnamefont {M.~M.}\ \bibnamefont
  {Wilde}},\ }\href {\doibase 10.1088/1751-8121/aa9da7} {\bibfield  {journal}
  {\bibinfo  {journal} {Journal of Physics A: Mathematical and Theoretical}\
  }\textbf {\bibinfo {volume} {51}},\ \bibinfo {pages} {035303} (\bibinfo
  {year} {2017})}\BibitemShut {NoStop}%
\bibitem [{\citenamefont {{Hayashi}}(2009)}]{Hayashi2009}%
  \BibitemOpen
  \bibfield  {author} {\bibinfo {author} {\bibfnamefont {M.}~\bibnamefont
  {{Hayashi}}},\ }\href {\doibase 10.1109/TIT.2009.2023726} {\bibfield
  {journal} {\bibinfo  {journal} {IEEE Transactions on Information Theory}\
  }\textbf {\bibinfo {volume} {55}},\ \bibinfo {pages} {3807} (\bibinfo {year}
  {2009})}\BibitemShut {NoStop}%
\bibitem [{\citenamefont {Duan}\ \emph {et~al.}(2009)\citenamefont {Duan},
  \citenamefont {Feng},\ and\ \citenamefont {Ying}}]{DFY2009}%
  \BibitemOpen
  \bibfield  {author} {\bibinfo {author} {\bibfnamefont {R.}~\bibnamefont
  {Duan}}, \bibinfo {author} {\bibfnamefont {Y.}~\bibnamefont {Feng}}, \ and\
  \bibinfo {author} {\bibfnamefont {M.}~\bibnamefont {Ying}},\ }\href {\doibase
  10.1103/PhysRevLett.103.210501} {\bibfield  {journal} {\bibinfo  {journal}
  {Phys. Rev. Lett.}\ }\textbf {\bibinfo {volume} {103}},\ \bibinfo {pages}
  {210501} (\bibinfo {year} {2009})}\BibitemShut {NoStop}%
\bibitem [{\citenamefont {Berta}\ \emph {et~al.}(2019)\citenamefont {Berta},
  \citenamefont {Hirche}, \citenamefont {Kaur},\ and\ \citenamefont
  {Wilde}}]{Berta-19}%
  \BibitemOpen
  \bibfield  {author} {\bibinfo {author} {\bibfnamefont {M.}~\bibnamefont
  {Berta}}, \bibinfo {author} {\bibfnamefont {C.}~\bibnamefont {Hirche}},
  \bibinfo {author} {\bibfnamefont {E.}~\bibnamefont {Kaur}}, \ and\ \bibinfo
  {author} {\bibfnamefont {M.~M.}\ \bibnamefont {Wilde}},\ }\href@noop {} {\
  (\bibinfo {year} {2019})},\ \Eprint {http://arxiv.org/abs/arXiv:1808.01498}
  {arXiv:1808.01498} \BibitemShut {NoStop}%
\bibitem [{\citenamefont {Navascu\'es}\ and\ \citenamefont
  {Garc\'{\i}a-Pintos}(2015)}]{NG2015}%
  \BibitemOpen
  \bibfield  {author} {\bibinfo {author} {\bibfnamefont {M.}~\bibnamefont
  {Navascu\'es}}\ and\ \bibinfo {author} {\bibfnamefont {L.~P.}\ \bibnamefont
  {Garc\'{\i}a-Pintos}},\ }\href {\doibase 10.1103/PhysRevLett.115.010405}
  {\bibfield  {journal} {\bibinfo  {journal} {Phys. Rev. Lett.}\ }\textbf
  {\bibinfo {volume} {115}},\ \bibinfo {pages} {010405} (\bibinfo {year}
  {2015})}\BibitemShut {NoStop}%
\bibitem [{\citenamefont {Faist}\ \emph {et~al.}(2019)\citenamefont {Faist},
  \citenamefont {Berta},\ and\ \citenamefont {Brand{\~a}o}}]{Faist-19}%
  \BibitemOpen
  \bibfield  {author} {\bibinfo {author} {\bibfnamefont {P.}~\bibnamefont
  {Faist}}, \bibinfo {author} {\bibfnamefont {M.}~\bibnamefont {Berta}}, \ and\
  \bibinfo {author} {\bibfnamefont {F.}~\bibnamefont {Brand{\~a}o}},\
  }\href@noop {} {\  (\bibinfo {year} {2019})},\ \Eprint
  {http://arxiv.org/abs/arXiv:1807.05610} {arXiv:1807.05610} \BibitemShut
  {NoStop}%
\bibitem [{\citenamefont {Watrous}(2009)}]{Watrous-2009}%
  \BibitemOpen
  \bibfield  {author} {\bibinfo {author} {\bibfnamefont {J.}~\bibnamefont
  {Watrous}},\ }\href {\doibase 10.4086/toc.2009.v005a011} {\bibfield
  {journal} {\bibinfo  {journal} {Theory of Computing}\ }\textbf {\bibinfo
  {volume} {5}},\ \bibinfo {pages} {217} (\bibinfo {year} {2009})}\BibitemShut
  {NoStop}%
\bibitem [{\citenamefont {Konig}\ \emph {et~al.}(2009)\citenamefont {Konig},
  \citenamefont {Renner},\ and\ \citenamefont {Schaffner}}]{Konig-2009a}%
  \BibitemOpen
  \bibfield  {author} {\bibinfo {author} {\bibfnamefont {R.}~\bibnamefont
  {Konig}}, \bibinfo {author} {\bibfnamefont {R.}~\bibnamefont {Renner}}, \
  and\ \bibinfo {author} {\bibfnamefont {C.}~\bibnamefont {Schaffner}},\ }\href
  {\doibase 10.1109/TIT.2009.2025545} {\bibfield  {journal} {\bibinfo
  {journal} {IEEE Transactions on Information Theory}\ }\textbf {\bibinfo
  {volume} {55}},\ \bibinfo {pages} {4337} (\bibinfo {year}
  {2009})}\BibitemShut {NoStop}%
\bibitem [{\citenamefont {Datta}\ \emph {et~al.}(2013)\citenamefont {Datta},
  \citenamefont {Mosonyi}, \citenamefont {Hsieh},\ and\ \citenamefont
  {Brand{\~a}o}}]{Datta-2013c}%
  \BibitemOpen
  \bibfield  {author} {\bibinfo {author} {\bibfnamefont {N.}~\bibnamefont
  {Datta}}, \bibinfo {author} {\bibfnamefont {M.}~\bibnamefont {Mosonyi}},
  \bibinfo {author} {\bibfnamefont {M.~H.}\ \bibnamefont {Hsieh}}, \ and\
  \bibinfo {author} {\bibfnamefont {F.~G. S.~L.}\ \bibnamefont {Brand{\~a}o}},\
  }\href {\doibase 10.1109/TIT.2013.2282160} {\bibfield  {journal} {\bibinfo
  {journal} {IEEE Transactions on Information Theory}\ }\textbf {\bibinfo
  {volume} {59}},\ \bibinfo {pages} {8014} (\bibinfo {year}
  {2013})}\BibitemShut {NoStop}%
\bibitem [{\citenamefont {Ogawa}\ and\ \citenamefont {Nagaoka}()}]{Ogawa-2000}%
  \BibitemOpen
  \bibfield  {author} {\bibinfo {author} {\bibfnamefont {T.}~\bibnamefont
  {Ogawa}}\ and\ \bibinfo {author} {\bibfnamefont {H.}~\bibnamefont
  {Nagaoka}},\ }\enquote {\bibinfo {title} {Strong converse and stein's lemma
  in quantum hypothesis testing},}\ in\ \href {\doibase
  10.1142/9789812563071_0003} {\emph {\bibinfo {booktitle} {Asymptotic Theory
  of Quantum Statistical Inference}}},\ pp.\ \bibinfo {pages}
  {28--42}\BibitemShut {NoStop}%
\bibitem [{\citenamefont {{Datta}}\ and\ \citenamefont
  {{Renner}}(2009)}]{Datta-2009}%
  \BibitemOpen
  \bibfield  {author} {\bibinfo {author} {\bibfnamefont {N.}~\bibnamefont
  {{Datta}}}\ and\ \bibinfo {author} {\bibfnamefont {R.}~\bibnamefont
  {{Renner}}},\ }\href {\doibase 10.1109/TIT.2009.2018340} {\bibfield
  {journal} {\bibinfo  {journal} {IEEE Transactions on Information Theory}\
  }\textbf {\bibinfo {volume} {55}},\ \bibinfo {pages} {2807} (\bibinfo {year}
  {2009})}\BibitemShut {NoStop}%
\bibitem [{\citenamefont {Audenaert}\ \emph {et~al.}(2007)\citenamefont
  {Audenaert}, \citenamefont {Calsamiglia}, \citenamefont {Mu\~noz Tapia},
  \citenamefont {Bagan}, \citenamefont {Masanes}, \citenamefont {Acin},\ and\
  \citenamefont {Verstraete}}]{Audenaert-2007}%
  \BibitemOpen
  \bibfield  {author} {\bibinfo {author} {\bibfnamefont {K.~M.~R.}\
  \bibnamefont {Audenaert}}, \bibinfo {author} {\bibfnamefont {J.}~\bibnamefont
  {Calsamiglia}}, \bibinfo {author} {\bibfnamefont {R.}~\bibnamefont {Mu\~noz
  Tapia}}, \bibinfo {author} {\bibfnamefont {E.}~\bibnamefont {Bagan}},
  \bibinfo {author} {\bibfnamefont {L.}~\bibnamefont {Masanes}}, \bibinfo
  {author} {\bibfnamefont {A.}~\bibnamefont {Acin}}, \ and\ \bibinfo {author}
  {\bibfnamefont {F.}~\bibnamefont {Verstraete}},\ }\href {\doibase
  10.1103/PhysRevLett.98.160501} {\bibfield  {journal} {\bibinfo  {journal}
  {Phys. Rev. Lett.}\ }\textbf {\bibinfo {volume} {98}},\ \bibinfo {pages}
  {160501} (\bibinfo {year} {2007})}\BibitemShut {NoStop}%
\end{thebibliography}%


\newpage
\onecolumngrid


\begin{center}\Large\bfseries
 Supplemental Material\\ How to quantify a dynamical quantum resource 
\end{center}

\section{A Zoo of relative entropies for a dynamical resource}

We introduce here six functions that generalize the relative entropy measure of 
static resources (i.e., states) to channels. We start with $D_\mf$ and $E_\mf$, and prove Theorem~\ref{properties}.
For any $\rho\in\mD(A)$, denote the relative entropy of resourceness by
\be
  D_\mf(\rho)\eqdef \min_{\sigma\in\mf(A)}D(\rho\|\sigma)\;.
\ee
We first show that both $D_\mf(\mN)$ and $E_\mf(\mN)$ reduces to this function when $\mN_{A\to B}$ is the replacement channel that always output a fixed state $\omega_B$. 

Indeed, in one direction we have
\ba
E_{\mf}(\mN)&=\min_{\mM\in\mf(A\to B)}\sup_{\rho_{RA}\in\mf(RA)}D\left(\rho_{R}\otimes\omega_B)\|\mM_{A\to B}(\rho_{RA})\right)\\
&\leq \min_{\mM=\gamma\in\mf(B)}\sup_{\rho_{RA}\in\mf(RA)}D\left(\rho_{R}\otimes\omega_B)\|\rho_{R}\otimes\gamma_B)\right)\\
&=\min_{\gamma\in\mf(B)}D\left(\omega_B\|\gamma_B\right)=D_\mf(\omega_B) ,
\ea
where the inequality follows from the restriction of the minimization over $\mf(A\to B)$ to minimization over replacement channels in $\mf(A\to B)$.

For the other direction,
\ba
E_{\mf}(\mN)&=\min_{\mM\in\mf(A\to B)}\sup_{\rho_{RA}\in\mf(RA)}D\left(\rho_{R}\otimes\omega_B)\|\mM_{A\to B}(\rho_{RA})\right)\\
&\geq \min_{\mM\in\mf(A\to B)}\max_{\rho_{A}\in\mf(A)}D\left(\omega_B\|\mM_{A\to B}(\rho_{A})\right)\\
&=\min_{\gamma\in\mf(B)}D\left(\omega_B\|\gamma_B\right)=D_\mf(\omega_B) ,
\ea
where the inequality follows from the monotonicity of the divergence under partial trace. This proves that $E_\mf(\mN)=D_\mf(\omega_B)$.
The proof that $D_\mf(\mN)\leq D_\mf(\omega_B)$ follows the exact same lines as above, and the proof that $D_\mf(\mN)\geq D_\mf(\omega_B)$ follows from the fact that $D_\mf(\mN)\geq E_\mf(\mN)$. Hence, we also have  $D_\mf(\mN)=D_\mf(\omega_B)$.

The function $E_{\mf}$ satisfies~\eqref{mono} for any $\Theta$ of the form~\eqref{super} with $\mE^{\post}\in\cptp(BE\to B')$ and $\mE^{\pree}\in\cptp(A'\to AE)$ both being completely RNG. To see it, note that
\ba
E_{\mf}\left(\Theta[\mN]\right)&=\min_{\Omega\in\mf(A'\to B')}\sup_{\rho\in\mf(RA')}D\big(\Theta[\mN_{A\to B}](\rho_{RA'})\big\|\Omega_{A'\to B'}(\rho_{RA'})\big)\\
&\leq \min_{\mM\in\mf(A\to B)}\sup_{\rho\in\mf(RA')}D\left(\Theta[\mN_{A\to B}](\rho_{RA'})\big\|\Theta[\mM_{A\to B}](\rho_{RA'})\right)\\
&= \min_{\mM\in\mf(A\to B)}\sup_{\rho\in\mf(RA')}D\left(\mE^{\post}_{BE\to B'}\circ\mN_{A\to B}\circ\mE^{\pree}_{A'\to AE}(\rho_{RA'})\big\|\mE^{\post}_{BE\to B'}\circ\mM_{A\to B}\circ\mE^{\pree}_{A'\to AE}(\rho_{RA'})\right)\\
&\leq \min_{\mM\in\mf(A\to B)}\sup_{\rho\in\mf(RA')}D\left(\mN_{A\to B}\circ\mE^{\pree}_{A'\to AE}(\rho_{RA'})\big\|\mM_{A\to B}\circ\mE^{\pree}_{A'\to AE}(\rho_{RA'})\right)\\
&\leq \min_{\mM\in\mf(A\to B)}\sup_{\rho\in\mf(RAE)}D\left(\mN_{A\to B}(\rho_{RAE})\big\|\mM_{A\to B}(\rho_{RAE})\right)\\
&=E_{\mf}(\mN_{A\to B}) .
\ea
The first inequality follows from the assumption that $\Theta$ is RNG so that $\Theta[\mM]\in\mf(A'\to B')$, the second inequality from data processing of $D$, 
and the third inequality from the assumption that $\mE^{\pree}_{A'\to AE}$ is completely RNG.

The faithfulness of $D_\mf$ follows directly from the definition. To prove the faithfulness of $E_\mf$ note that if $E_\mf(\mN)=0$ for some $\mN\in\cptp(A\to B)$ then from the Klein's inequality (applied to the relative entropy) for all $\rho\in\mf(RA)$ there exists $\mM\in\mf(A\to B)$ such that
\be
\mN_{A\to B}(\rho_{RA})=\mM_{A\to B}(\rho_{RA})\in\mf(RB)\;.
\ee
Therefore, $\mN$ must be completely RNG. Moreover, taking $|R|=|A|$, we conclude that if $\mf(RA)$ contains a pure state with full Schmidt rank then the equation above (with $\rho_{RA}$ being that pure state) implies that $\mN=\mM$; i.e. $\mN\in\mf(A\to B)$.

\subsection{Other relative entropies of a dynamical resource}

In addition to $D_\mf$ and $E_\mf$, there are other functionals that extend 
the relative entropy of a resource from states to channels. Here we discuss 
four additional generalizations.

\subsubsection{Two state-based measures}

There are two resource monotones that involve no optimization over channels in $\mf(A\to B)$, but only optimization over states.
They were  introduced very recently in~\cite{LY2019,LW2019}.
Let $\mN\in\cptp(A\to B)$ and define
\begin{align}
  R_{\mf}(\mN_{A\to B}) &\eqdef \sup_{\sigma\in\mD(RA)}
                                 \Big(D_{\mf}(\mN_{A\to B}(\sigma_{RA}))-D_{\mf}(\sigma_{RA})\Big), \\
  \tR_\mf(\mN_{A\to B}) &\eqdef \sup_{\sigma\in\mf(RA)}
                                 D_{\mf}(\mN_{A\to B}(\sigma_{RA})).
\end{align}
Note that $\tR_\mf$ can be obtained from the expression above for $R_\mf$, by restricting the supremum over $\sigma\in\mD(RA)$ to $\sigma\in\mf(RA)$. Hence, we always have $\tR_\mf(\mN)\leq R_\mf(\mN)$.
We now show that both $R_{\mf}$ and $\tR_\mf$ behave monotonically under completely RNG superchannels.
\begin{lemma}
Let $\Theta:\cptp(A\to B)\to\cptp(A'\to B')$ be a superchannel defined by
\be
  \Theta[\mN_{A\to B}]\eqdef \mE^{\post}_{BE\to B'}\circ\mN_{A\to B}\circ\mE^{\pree}_{A'\to AE} ,
\ee
with $\mE^{\pree}\in\cptp(A'\to AE)$ and $\mE^{\post}\in\cptp(BE\to B')$ being completely RNG.
Then,
\be
  R_{\mf}(\Theta[\mN])\leq R_{\mf}(\mN)\quad;\quad\tR_{\mf}(\Theta[\mN])\leq \tR_{\mf}(\mN)\;.
\ee
\end{lemma}
\begin{proof}
From the definitions we have:
\ba
R_{\mf}(\Theta[\mN_{A\to B}])&=\sup_{\sigma\in\mD(R'A')}\Big(D_{\mf}\left(\mE^{\post}_{BE\to B'}\circ\mN_{A\to B}\circ\mE^{\pree}_{A'\to AE}(\sigma_{R'A'})\right)-D_{\mf}(\sigma_{R'A'})\Big)\\
&\leq \sup_{\sigma\in\mD(R'A')}\Big(D_{\mf}\left(\mN_{A\to B}\circ\mE^{\pree}_{A'\to AE}(\sigma_{R'A'})\right)-D_{\mf}(\sigma_{R'A'})\Big)\\
&\leq \sup_{\sigma\in\mD(R'A')}\Big(D_{\mf}\left(\mN_{A\to B}\circ\mE^{\pree}_{A'\to AE}(\sigma_{R'A'})\right)-D_{\mf}\left(\mE^{\pree}_{A'\to AE}(\sigma_{R'A'})\right)\Big)\\
&\leq \sup_{\rho\in\mD(R'AE)}\Big(D_{\mf}\left(\mN_{A\to B}(\rho_{R'AE})\right)-D_{\mf}\left(\rho_{R'AE})\right)\Big)\\
&=\sup_{\rho\in\mD(RA)}\Big(D_{\mf}\left(\mN_{A\to B}(\rho_{RA})\right)-D_{\mf}\left(\rho_{RA})\right)\Big)\\
&=R_{\mf}(\mN_{A\to B})\;.
\ea
In the first inequality we used the fact that $D_{\mf}$ is monotonic under the RNG map $\id_R\otimes\mE^{\post}_{BE\to B'}$ (recall that we assume that $\mE^{\post}$ is completely RNG). Similarly, for the second inequality we used the monotonicity of $D_{\mf}$ under $\id_R\otimes\mE^{\pree}_{A'\to AE}$.
Finally, we substituted an arbitrary state $\rho_{R'AE}$ instead of $\mE^{\pree}_{A'\to AE}(\sigma_{R'A'})$ and set $R\equiv R'E$. The proof of the monotonicity of $\tR$ follows the exact same lines by replacing everywhere the set $\mD(RA)$ with $\mf(RA)$.
\end{proof}

The next lemma shows that $R_{\mf}$ and $\tR_\mf$ are indeed generalizations of the relative entropy of a resource.
\begin{lemma}
In the case that $\mN_{A\to B}=\omega_B$ is a replacement channel, it holds
\be
  R_{\mf}(\mN_{A\to B})=\tR_{\mf}(\mN_{A\to B})=D_\mf(\omega_B)\;.
\ee
\end{lemma}
\begin{proof}
We have
\ba
R_{\mf}(\mN_{A\to B})&=\sup_{\rho\in\mD(RA)}\sup_{\sigma\in\mf(RA)}\Big(D_\mf(\rho_{R}\otimes\omega_B)-D(\rho_{RA}\|\sigma_{RA})\Big)\\
&=\sup_{\eta^R\in\mD(R)}\sup_{\sigma\in\mf(RA)}\Big(D_\mf(\rho_{R}\otimes\omega_B)-\inf_{\{\rho\in\mD(RA)\;:\;\rho^R=\eta^R\}}D(\rho_{RA}\|\sigma_{RA})\Big)\\
&=\sup_{\eta^R\in\mD(R)}\Big(D_\mf(\rho_{R}\otimes\omega_B)-\inf_{\sigma\in\mf(RA)}\inf_{\{\rho\in\mD(RA)\;:\;\rho^R=\eta^R\}}D(\rho_{RA}\|\sigma_{RA})\Big) .
\ea
Now, observe that from the data processing inequality
\be
  \label{eqq}
  \inf_{\sigma\in\mf(RA)}\inf_{\{\rho\in\mD(RA)\;:\;\rho^R=\eta^R\}}D(\rho_{RA}\|\sigma_{RA})
        \geq \inf_{\sigma\in\mf(R)}D(\eta_{R}\|\sigma_{R})=D_\mf(\eta_R) ,
\ee
where the inequality above is in fact an equality as can be seen by taking $\rho_{RA}=\eta_R\otimes\sigma_A$ and $\sigma_{RA}=\sigma_R\otimes\sigma_A$. Similarly, by using the subadditivity of $D_\mf$, we get that
\be
  D_\mf(\eta_{R}\otimes\omega_B)\leq D_\mf(\eta_R)+D_\mf(\omega_B)\;,
\ee
so that together with~\eqref{eqq} (with the inequality replaced with equality) we conclude
\be
  R_{\mf}(\mN_{A\to B})\leq D_\mf(\omega_B).
\ee
To get the other direction, note that restricting $\eta_R$ to $\mf(R)$ gives
\ba
R_{\mf}(\mN_{A\to B})&\geq\sup_{\eta_R\in\mf(R)}\Big(D_\mf(\eta_{R}\otimes\omega_B)-\inf_{\sigma\in\mf(RA)}\inf_{\{\rho\in\mD(RA)\;:\;\rho_R=\eta_R\}}D(\rho_{RA}\|\sigma_{RA})\Big)\\
&=\sup_{\eta_R\in\mf(R)}D_\mf(\eta_{R}\otimes\omega_B)\\
&\geq \sup_{\eta_R\in\mf(R)}D_\mf(\omega_B)\\
&=D_\mf(\omega_B).
\ea
This completes the proof that $R_\mf(\mN_{A\to B})=D_\mf(\omega_B)$. The proof that $\tR_\mf(\mN_{A\to B})=D_\mf(\omega_B)$ follows along similar lines.
\end{proof}

\subsubsection{Two measures that are based on the amortized divergence}

There is another way to extend a divergence $D$ to channels. It was introduced in~\cite{Berta-19} under the name \emph{amortized divergence}. It is defined as 
\be
  D^{\mA}(\mN\|\mM) \eqdef \sup_{\rho,\sigma\in\mD(RA)}
          D\left(\mN_{A\to B}(\rho_{RA})\|\mM_{A\to B}(\sigma_{RA})\right)-D(\rho_{RA}\|\sigma_{RA}) .
\ee
Like $D(\mN\|\mM)$, also $D^{\mA}(\mN\|\mM)$ satisfies the generalized data processing inequality~\cite{Berta-19}. That is, for any superchannel $\Theta: \cptp(A\to B)\to\cptp(A'\to B')$,
\be
  D^{\mA}(\Theta[\mN]\|\Theta[\mM])\leq D^{\mA}(\mN\|\mM).
\ee
Define two functionals
\begin{align}
  D^{\mA}_{\mf}(\mN) &\eqdef \min_{\mM\in\mf(A\to B)}D^{\mA}(\mN\|\mM), \\
  E^{\mA}_{\mf}(\mN) &\eqdef \min_{\mM\in\mf(A\to B)}\sup_{\rho,\sigma\in\mf(RA)}
             D\left(\mN_{A\to B}(\rho_{RA})\|\mM_{A\to B}(\sigma_{RA})\right)-D(\rho_{RA}\|\sigma_{RA}). 
\end{align}
Note that for any $\mN\in\cptp(A\to B)$, we have by definition 
\be
  D_\mf(\mN)\leq D_{\mf}^{\mA}(\mN)\quad;\quad E_\mf(\mN)\leq E^{\mA}_{\mf}(\mN)\;.
\ee
Therefore, the faithfulness of these functions follows from that of $D_\mf$ and $E_{\mf}$. The next lemma shows that they behave monotonically under completely RNG superchannels.
\begin{lemma}
Let $\Theta:\cptp(A\to B)\to\cptp(A'\to B')$ be a superchannel defined by
\be
  \Theta[\mN_{A\to B}]\eqdef \mE^{\post}_{BE\to B'}\circ\mN_{A\to B}\circ\mE^{\pree}_{A'\to AE},
\ee
with $\mE^{\pree}\in\cptp(A'\to AE)$ and $\mE^{\post}\in\cptp(BE\to B')$ being completely RNG.
Then,
\be
  D_\mf^{\mA}(\Theta[\mN])
    \leq D_\mf^{\mA}(\mN)\quad\text{and}\quad E^{\mA}_{\mf}(\Theta[\mN])\leq E^{\mA}_{\mf}(\mN)\;.
\ee
\end{lemma}

\begin{proof}
The monotonicity of $D_{\mf}^{\mA}$ follows from the data processing inequality of the amortized divergence. Indeed,
\ba
  D_\mf^{\mA}(\Theta[\mN])
    &=    \min_{\Omega\in\mf(A'\to B')}D^{\mA}(\Theta[\mN_{A\to B}]\|\Omega_{A'\to B'})\\
    &\leq \min_{\mM\in\mf(A\to B)}D^{\mA}(\Theta[\mN_{A\to B}]\|\Theta[\mM_{A\to B}])\\
    &\leq \min_{\mM\in\mf(A\to B)}D^{\mA}(\mN_{A\to B}\|\mM_{A\to B})=D_{\mf}^{\mA}(\mN) .
\ea
The monotonicity of $E_{\mf}^{\mA}$ is proved as follows:
\ba\nonumber
E_{\mf}^{\mA}&\left(\Theta[\mN]\right)
 =    \min_{\Omega\in\mf(A'\to B')}\sup_{\rho,\sigma\in\mf(RA')}
          D\big(\Theta[\mN_{A\to B}](\rho_{RA'})\big\| \Omega_{A'\to B'}(\sigma_{RA'})\big) 
           - D(\rho_{RA'}\|\sigma_{RA'}) \\
&\leq \min_{\mM\in\mf(A\to B)} \sup_{\rho,\sigma\in\mf(RA')}\left(\Theta[\mN_{A\to B}](\rho_{RA'})\big\|\Theta[\mM_{A\to B}](\sigma_{RA'})\right)-D(\rho_{RA'}\|\sigma_{RA'})\\
&\leq \min_{\mM\in\mf(A\to B)} \sup_{\rho,\sigma\in\mf(RA')}D\left(\mN_{A\to B}\circ\mE^{\pree}_{A'\to AE}(\rho_{RA'})\big\|\mM_{A\to B}\circ\mE^{\pree}_{A'\to AE}(\sigma_{RA'})\right)-D(\rho_{RA'}\|\sigma_{RA'})\\
&\leq \min_{\mM\in\mf(A\to B)} \sup_{\rho,\sigma\in\mf(RA')}D\left(\mN_{A\to B}\circ\mE^{\pree}_{A'\to AE}(\rho_{RA'})\big\|\mM_{A\to B}\circ\mE^{\pree}_{A'\to AE}(\sigma_{RA'})\right)-D\big(\mE^{\pree}_{A'\to AE}(\rho_{RA'})\big\|\mE^{\pree}_{A'\to AE}(\sigma_{RA'})\big)\\
&\leq \min_{\mM\in\mf(A\to B)}\sup_{\rho,\sigma\in\mf(RAE)}D\left(\mN_{A\to B}(\rho_{RAE})\big\|\mM_{A\to B}(\sigma_{RAE})\right)-D(\rho_{RAE}\|\sigma_{RAE})\\
&=    E_{\mf}^{\mA}(\mN_{A\to B}).
\ea
The first inequality follows from the assumption that $\Theta$ is RNG, the second and third inequalities follow from data processing inequality of $D$, and
the fourth inequality follows from the assumption that $\mE^{\pree}$ is completely RNG.
\end{proof}

Finally, we show that for a replacement channel $\mN_{A\to B}$ that outputs a 
fixed state $\omega_B$,
\be
  D_\mf^\mA(\mN)=E_\mf^{\mA}(\mN)=D_\mf(\omega)\;.
\ee
Indeed,
\ba
D_\mf^\mA(\mN)&=\min_{\mM\in\mf(A\to B)}\sup_{\rho,\sigma\in\mf(RA)}D\left(\rho_{R}\otimes\omega_B)\|\mM_{A\to B}(\sigma_{RA})\right)-D(\rho_{RA}\|\sigma_{RA})\\
&\leq \min_{\mM=\gamma\in\mf(B)}\sup_{\rho,\sigma\in\mf(RA)}D\left(\rho_{R}\otimes\omega_B\|\sigma_{R}\otimes\gamma_B\right)-D(\rho_{RA}\|\sigma_{RA})\\
&\leq \min_{\mM=\gamma\in\mf(B)}\sup_{\rho,\sigma\in\mf(RA)}D\left(\rho_{R}\otimes\omega_B\|\sigma_{R}\otimes\gamma_B\right)-D(\rho_{R}\|\sigma_{R})\\
&=\min_{\gamma\in\mf(B)}D\left(\omega_B\|\gamma_B\right)\\
&=D_\mf(\omega_B),
\ea
where the first inequality follows from the restriction of the minimization over $\mf(A\to B)$ to minimization over replacement channels in $\mf(A\to B)$. The second inequality follows from data processing of the relative entropy $D$, and the following equality follows from the additivity of the relative entropy.
To prove the other direction, note that $D_\mf^\mA(\mN)\geq D_\mf(\mN)\geq D_\mf(\omega_B)$. 
Hence, $D_\mf^\mA(\mN)=D_\mf(\omega_B)$.

For the proof the $E_\mf^\mA(\mN)=D_\mf(\omega_B)$, note that $E_\mf^\mA(\mN)\leq D_\mf^\mA(\mN)=D_\mf(\omega_B)$, and for the other direction, $E_\mf^\mA(\mN)\geq E_\mf(\mN)=D_\mf(\omega_B)$. This proves that also $E_\mf^\mA(\mN)=D_\mf(\omega_B)$.

\subsubsection{The form of the monotones in the resource theory of thermodynamics}
Since in the QRT of athermality $\mf(A)$ consists of only one free state, namely 
the Gibbs state at fixed temperature, some of the relative entropies discussed above 
take simple forms. Here we discuss a few of them.
Let the set of free states consists of a single Gibbs state $\mf(A)=\{\gamma_A\}$ 
and $\mf(B)=\{\gamma_B\}$. Then,
\ba
  E_\mf(\mN)
    &= \sup_{|R|\in\mbb{N}}\inf_{\mM\in\mf(A\to B)}
        D\left(\mN_{A\to B}(\gamma_{R}\otimes\gamma_A)\|\mM_{A\to B}(\gamma_{R}\otimes\gamma_A)\right)\\
    &= \inf_{\mM\in\mf(A\to B)}D\left(\mN_{A\to B}(\gamma_A)\|\mM_{A\to B}(\gamma_A)\right)\\
    &= D\left(\mN_{A\to B}(\gamma_A)\|\gamma_B)\right),
\ea
which is the Gibbs free energy of the state $\mN_{A\to B}(\gamma_A)$. Note that this is 
also the value of $\tR_\mf$ so that in the QRT of athermality we have the collapse
\be
  E_\mf(\mN)=\tR_\mf(\mN)=F\left(\mN_{A\to B}(\gamma_A)\right),
\ee
where $F$ stands for the free energy.

Finally, we show that $R_\mf$ reduces to the thermodynamic capacity in the QRT of athermality.
\begin{lemma}
  In the thermodynamic case, in which the set of free states consists of a single 
  Gibbs state $\mf(A)=\{\gamma_A\}$ and $\mf(B)=\{\gamma_B\}$, we have:
  \be
    R_{\mf}(\mN_{A\to B})
      = \sup_{\sigma\in\mD(A)}\Big(D(\mN_{A\to B}(\sigma_{A})\big\|\gamma_B)
                                   -D(\sigma_{A}\big\|\gamma_{A})\Big)
      \equiv T(\mN_{A\to B}),
  \ee
  where $T(\mN_{A\to B})$ is the thermodynamic capacity of the channel as defined 
  in~\cite{NG2015} (see also~\cite{Faist-19}, where ot bis shown that the same 
  quantity is the work cost of implementing $(\mN_{A\to B}$ using Gibbs-preserving
  operations).
\end{lemma}
\begin{proof}
In this case,
\be
R_{\mf}(\mN_{A\to B})\eqdef\sup_{\sigma\in\mD(RA)}\Big(D(\mN_{A\to B}(\sigma_{RA})\big\|\gamma_R\otimes\gamma_B)-D(\sigma_{RA}\big\|\gamma_R\otimes\gamma_{A})\Big)
\ee
Now, note that
\begin{align}
D\Big(\mN_{A\to B}(\sigma_{RA})&\big\|\gamma_R\otimes\gamma_B\Big)-D\Big(\sigma_{RA}\big\|\gamma_R\otimes\gamma_{A}\Big)\\
  &= - H\left(\mN_{A\to B}(\sigma_{RA})\right)
      - \tr\left[\mN_{A\to B}(\sigma_{RA})\log(\gamma_R\otimes\gamma_B)\right]
      + H(\sigma_{RA})+\tr\left[\sigma_{RA}\log(\gamma_R\otimes\gamma_B)\right]\\
  &= D\left(\mN_{A\to B}(\sigma_{A})\big\|\gamma_B\right)
     - D(\sigma_{A}\big\|\gamma_{A})+H(R|A)_{\sigma}-H(R|B)_{\mN_{A\to B}(\sigma_{RA})} .
\end{align}
Furthermore, from the data processing inequality we have
\be
  H(R|A)_{\sigma_{RA}}\leq H(R|B)_{\mN_{A\to B}(\sigma_{RA})},
\ee
with equality if $\sigma_{RA}=\sigma_R\otimes\sigma_A$. This completes the proof.
\end{proof}

\section{Minimax Theorem for the relative entropy}

Consider a distance parameter $d:\mD(A)\times\mD(A)\to\mbb{R}_{+}$ on states that is non-negative and contractive (monotone) 
under CPTP maps. Let $\mS(RA)$ be a convex set of density matrices. We will take here $\mS(RA)=\mD(RA)$ or $\mS(RA)=\mf(RA)$. For a channel $\mN\in\cptp(A\to B)$, and a QRT $\mathfrak{F}$, define
\begin{align}
  \underline{d}_{\mathfrak{F},\mS}(\mN) &:= \sup_{\rho\in\mS(RA)} \inf_{\mM\in\mathfrak{F}(A\rightarrow B)} 
                                              d\bigl( \mN_{A\to B}(\rho_{RA}), \mM_{A\to B}(\rho_{RA}) \bigr), \\
  \overline{d}_{\mathfrak{F},\mS}(\mN)  &:= \inf_{\mM\in\mathfrak{F}(A\rightarrow B)} \sup_{\rho\in\mS(RA)}
                                              d\bigl( \mN_{A\to B}(\rho_{RA}), \mM_{A\to B}(\rho_{RA}) \bigr).
\end{align}

By general principles (max-min inequality), $\underline{d}_{\mathfrak{F},\mS}(\mN)  \leq  \overline{d}_{\mathfrak{F},\mS}(\mN)$,
and we will show equality under mild assumptions on $d$ and the free channels.
Concretely, assume that $d$ is jointly concave under orthogonally flagged mixtures:
This means that for any two families $\{\rho_x\}$ and $\{\sigma_x\}$ of states, and
any probability distribution $\{p_x\}$,
\begin{equation}
  \label{eq:concave2}
  d\left( \sum_x p_x \rho_x\ox\proj{x}, \sum_x p_x \sigma_x\ox\proj{x} \right) 
                                              \geq \sum_x p_x d(\rho_x,\sigma_x),
\end{equation}
where $\{\ket{x}\}$ is an orthonormal basis of an auxiliary system.
This for example holds with equality for the trace distance, relative entropy, and all the R\'enyi divergences.

For the case that, $\mS=\mf$, we will assume (in addition to convexity) that there exists a finite dimensional system $R$ such that $\mf(R)$ contains at least two orthonormal pure states. Since $\mf$ also admits the tensor product structure, this means that there exists a system $R'$ containing any finite number of orthonormal pure states. Hence, combining it with the convexity property, if $\{\rho^i\}\subset\mf(A)$ and $\{p_i\}$ is a probability distribution, then there exists a system $R$ and orthonormal set of pure states $\{|i\lr i|\}\subset\mf(R)$ such that $\sum_ip_i|i\lr i|_{R}\otimes\rho^i_A\in\mf(RA)$.

\subsection{Proof of Theorem~\ref{minmax}}
\begin{theorem*}
  \label{lemma:minimax}
  For a distance measure satisfying Eq. (\ref{eq:concave2}), and assuming that
  $\mathfrak{F}(A\rightarrow B)$ is convex (and satisfies the property above), and that $d$ is convex in the second
  argument, it holds
  $\underline{d}_{\mathfrak{F},\mS}(\mN) = \overline{d}_{\mathfrak{F},\mS}(\mN)$.
\end{theorem*}
\begin{proof}
We have automatically ``$\leq$'', so we will focus on proving ``$\geq$''.
Fix $R$ for the moment to be a finite-dimensional system. Since $\mf(A\to B)$ is a convex closed set,
any channel $\mM\in\mf(A\to B)$ can be expressed as a convex combination $\mM=\sum_jq_j\mM^j$, where each $\mM^j$ is an extreme channel of $\mf(A\to B)$. Similarly, since $\mS(A)$ is convex, every density matrix $\rho_{RA}$ can be expressed as a convex combination $\rho=\sum_ip_i\rho^i$, where each $\rho_i$ is an extreme state of $\mS(A)$. This means that the optimization over all channels and states in $\mf(A\to B)$ and $\mS(A)$ can be replaced with optimizations over the probability distributions $\{q_j\}$ and $\{p_i\}$. With this in mind we have
\[\begin{split}
  \sup_{\rho\in\mS(RA)} 
  \inf_{\mM\in\mathfrak{F}(A\rightarrow B)} 
         d\bigl( \mN_{A\to B}(\rho_{RA}), \mM_{A\to B}(\rho_{RA}) \bigr) 
         &=     \sup_{\rho\in\mS(RA)} \inf_{\{q_j\}}
                                 \sum_j q_j d\bigl( \mN_{A\to B}(\rho_{RA}), \mM_{A\to B}^j(\rho_{RA}) \bigr) \\
         &\leq \inf_{\{q_j\}} \sup_{\rho\in\mS(RA)} 
                                 \sum_j q_j d\bigl( \mN_{A\to B}(\rho_{RA}), \mM_{A\to B}^j(\rho_{RA}) \bigr) \\
         &=    \inf_{\{q_j\}} 
               \sup_{\{p_i\}} 
                                 \sum_{ij} p_iq_j d\bigl( \mN_{A\to B}(\rho_{RA}^i), \mM_{A\to B}^j(\rho_{RA}^i) \bigr) \\
         &=    \sup_{\{p_i\}} 
               \inf_{\{q_j\}} 
                                 \sum_{ij} p_iq_j d\bigl(  \mN_{A\to B}(\rho_{RA}^i), \mM_{A\to B}^j(\rho_{RA}^i)  \bigr) \\
         &\leq \sup_{\{p_i\}} 
               \inf_{\{q_j\}} 
                 \sum_j q_j d\bigl(  \mN_{A\to B}(\overline\rho_{R'RA}), \mM_{A\to B}^j(\overline\rho_{R'RA})  \bigr) \\
         &\leq \sup_{\rho\in\mS({RR'A})} 
               \inf_{\{q_j\}} 
                                   \sum_j q_j d\bigl(  \mN_{A\to B}(\rho_{R'RA}), \mM_{A\to B}^j(\rho_{R'RA})  \bigr) \\
         &=    \sup_{\rho\in\mS({RR'A})} \inf_{\mM\in\mathfrak{F}(A\rightarrow B)} 
                                      d\bigl( \mN_{A\to B}(\rho_{R'RA}), \mM_{A\to B}(\rho_{R'RA}) \bigr).
\end{split}\]
The first line is because the optimal ensemble $\{q_j,\mM_j\}$ of free channels
will be a point mass on a single optimal channel; the second is due to the general
minimax inequality; the third is by the same principle as the first; the fourth
line is due to von Neumann's minimax theorem, noting that the domains of
optimization are both convex, and the objective function is linear in either
variable; in the fifth,  we use the joint concavity
with $\overline\rho = \sum_i p_i \proj{i}^{R'}\ox\rho_i^{RA}$; in the sixth line,
we enlarge the maximization to arbitrary states on $\mS(RR'A)$; and in the seventh
we use once more the convex combination principle from lines 1 and 3.

Now, taking the supremum over auxiliary systems $R$, both the l.h.s. and
the r.h.s. yield $\underline{d}_{\mathfrak{F},\mS}(\mN)$, and all inequalities
above turn into equalities. In particular, $\underline{d}_{\mathfrak{F},\mS}(\mN)$ equals
the term in the second line, which evaluates to 
\[\begin{split}
  \underline{d}_{\mathfrak{F},\mS}(\mN)&= \inf_{\{q_j\}} \sup_{\rho\in\mS(RA)} 
                                 \sum_j q_j d\bigl( \mN_{A\to B}(\rho_{RA}), \mM_{A\to B}^j(\rho_{RA}) \bigr) \\
         &= \inf_{\mM\in\mathfrak{F}(A\rightarrow B)} \sup_{\rho\in\mS(RA)}
                                        d\bigl(  \mN_{A\to B}(\rho_{RA}), \mM_{A\to B}(\rho_{RA}) \bigr) 
          = \overline{d}_{\mathfrak{F},\mS}(\mN),
\end{split}\]
because the convexity of $\mathfrak{F}$ and $d$.
\end{proof}

\medskip
Without the convexity of $\mathfrak{F}$ and of $d$ in the second argument,
there is still something we can do: simply define
\[
  \widetilde{d}_{\mathfrak{F},\mS}(\mN) :=
        \inf_{\{q_j\}} \sup_{\rho\in\mS(RA)} 
                                 \sum_j q_j d\bigl( \mN_{A\to B}(\rho_{RA}), \mM_{A\to B}^j(\rho_{RA}) \bigr),
\]
then the above proof shows
\begin{lemma}
  \label{lemma:minimax-2}
  For a distance measure satisfying Eq. (\ref{eq:concave2}), it holds
  $\underline{d}_{\mathfrak{F},\mS}(\mN) = \widetilde{d}_{\mathfrak{F},\mS}(\mN)$.
  \qed
\end{lemma}

\section{Asymptotic continuity}

In this section we prove that the functions $D_\mf$ and $E_\mf$ are asymptotically 
continuous. For this purpose, we first need to check if they are bounded from above. 
Since $E_\mf\leq D_\mf$ it is sufficient to bound $D_\mf$.
Now, recall that we only consider here QRTs that admits the tensor product structure, so that the replacement channels $\mM_{\sigma}\in\mf(A\to B)$ for any $\sigma\in\mf(B)$. Hence,
\ba
D_\mf(\mN)&\leq \min_{\sigma\in\mf(B)}\max_{\varphi\in\mD(RA)}D\left(\mN_{A\to B}(\varphi_{RA})\|\varphi_{R}\otimes\sigma_B)\right)\\
&=\min_{\sigma\in\mf(B)}\max_{\varphi\in\mD(RA)}\Big\{D\left(\mN_{A\to B}(\varphi_{RA})\|\varphi_{R}\otimes\mN_{A\to B}(\varphi_A)\right)
+D\left(\mN_{A\to B}(\varphi_{A})\|\sigma_B\right)\Big\}\\
&=\min_{\sigma\in\mf(B)}\max_{\varphi\in\mD(RA)}\Big\{H(R:B)_{\mN_{A\to B}(\varphi_{RA})}+D\left(\mN_{A\to B}(\varphi_{A})\|\sigma_B\right)\Big\}\\
&\leq\log(|AB|)+ \min_{\sigma\in\mf(B)}\max_{\varphi_{A}}D\left(\mN_{A\to B}(\varphi_{A})\|\sigma_B\right)\\
&\leq \log(|AB|)+ \min_{\sigma\in\mf(B)}\max_{\varphi_{A}}-\tr\left[\mN_{A\to B}(\varphi_{A})\log\sigma_B\right]\\
&\leq \log(|AB|)+ \min_{\sigma\in\mf(B)}\log\|\sigma_B^{-1}\|_{\infty},
\ea
where we assumed w.l.o.g. $R\cong A$, and the second line follows from the 
following triangle equality property of the relative entropy 
\be
  D(\rho_{AB}\|\rho_A\otimes\tau_B)=D(\rho_{AB}\|\rho_A\otimes\rho_B)+D(\rho_{B}\|\tau_B)\;.
\ee
We will therefore assume that $\mf(B)$ contain a full rank state to get that $D_\mf(\mN)$ is bounded. For example, if $\mf(B)$ contains the maximally mixed (uniform) state $I_B/|B|$ then
\be
  D_\mf(\mN_{A\to B})\leq \log\bigl(|B|^2|A|\bigr).
\ee

\subsection{Proof of Theorem~\ref{acn}}

\subsubsection{Weaker Version}

This version only applies to $D_\mf$.

\begin{theorem*}
Let $\mf$ be a convex QRT such that
\be
  \kappa\eqdef\max_{\mN\in\cptp(A\to B)}D_\mf(\mN)\leq c \log|AB|
\ee
for some constant $c\in\mbb{R}_{+}$ independent of dimensions.
Then, $D_\mf$ is asymptotically continuous. In particular, for two channels 
$\mN,\mM\in\cptp(A\to B)$ and with 
$\epsilon\eqdef\frac{1}{2}\|\mN_{A\to B}-\mM_{A\to B}\|_{\diamond}$, we have
\be
  \bigl|D_{\mf}(\mN_{A\to B})-D_{\mf}(\mM_{A\to B})\bigr|
      \leq (1+\epsilon)h\left(\frac{\epsilon}{1+\epsilon}\right)+\epsilon\kappa, 
\ee
where $h(x)\eqdef -x\log x-(1-x)\log(1-x)$.
\end{theorem*}
\begin{proof}
We will be using the notation $J_{AB}^{\mN}=\sum_{x,y}|x\lr y|_{A}\otimes\mN_{A\to B}(|x\lr y|_A)$ for the Choi matrix of a quantum channel $\mN\in\cptp(A\to B)$. The diamond norm has been shown 
to be an SDP~\cite{Watrous-2009}, and in particular can be written as
\be
\|\mN-\mM\|_{\diamond}=\max_\varphi\left\|\mN_{\tA\to B}(\varphi_{A\tA})-\mM_{\tA\to B}(\varphi_{A\tA})\right\|_1=2\min_{\omega_{AB}\geq 0\;;\;\omega_{AB}\geq J^{\mN-\mM}_{AB}}\|\omega_A\|_{\infty}.
\ee
Note that there is always an optimal $\omega_{AB}$ such that $\omega_A=\epsilon I_A$. 
Therefore, the diamond norm can also be expressed as
\begin{align}
\frac{1}{2}\|\mN-\mM\|_{\diamond}
&=\min_{\omega_{AB}\geq 0\;;\;\omega_{AB}\geq J^{\mN-\mM}_{AB}}\|\omega_A\|_{\infty}\\
&=\min\Big\{\lambda : \lambda J_{AB}^{\mE}\geq J_{AB}^{\mN-\mM}\;\;;\;\;\mE\in\cptp(A\to B)\Big\}\\
&=\min\Big\{\lambda : \lambda \mE\geq \mN-\mM\;\;;\;\;\mE\in\cptp(A\to B)\Big\}\\
&=\min_{\mE\in\cptp(A\to B)}2^{D_{\max}(\mN-\mM\|\mE)}\equiv 2^{LR_{\cptp}(\mN-\mM)}.
\end{align}
That is, the diamond norm can be viewed as the $2^{D_{\max}}$ distance of 
$\mN-\mM$ to the set of all quantum channels $\cptp(A\to B)$. We point out that the entropy associated with $D_{\max}$ is the min-entropy~\cite{Konig-2009a,Datta-2013c}, and a direct relation between the min-entropy and the diamond norm of channels have been shown in~\cite{}).

Define the CPTP maps $\Delta_{\pm}$ in terms of the optimal matrix $\omega_{AB}$ as (recall that $\omega_A=\epsilon I_A$)
\be
J^{\Delta^{+}}_{AB}\eqdef\epsilon^{-1}\omega_{AB}\quad\text{and}\quad J^{\Delta^{-}}_{AB}\eqdef J^{\Delta^{+}}_{AB}-\epsilon^{-1}J^{\mN-\mM}_{AB}\;.
\ee
Note that
\be
  J^{\mN-\mM}_{AB}=\epsilon J^{\Delta^{+}-\Delta^{-}}_{AB}.
\ee
Dividing both sides by $1+\epsilon$  gives
\be
  \Omega_{A\to B} \eqdef \frac{1}{1+\epsilon}\mN+\frac{\epsilon}{1+\epsilon}\Delta^{-}
                  =      \frac{1}{1+\epsilon}\mM+\frac{\epsilon}{1+\epsilon}\Delta^{+}\;.
\ee
Define also
$$\mE_{A\to B}\eqdef\frac{1}{1+\epsilon}\mE^{1}_{A\to B}+\frac{\epsilon}{1+\epsilon}\mE^2_{A\to B}\;,$$ 
where $\mE^1_{A\to B},\mE^2_{A\to B}\in\mf(A\to B)$ are free quantum channels.
With this at hand, for any channels as above we have 
\be
  \label{1111}
  D\left(\Omega_{A\to B}(\varphi_{RA})\big\|\mE_{A\to B}(\varphi_{RA})\right)
    \leq \frac{1}{1+\epsilon} D\left(\mN_{A\to B}(\varphi_{RA})\big\|\mE_{A\to B}^1(\varphi_{RA})\right)
            + \frac{\epsilon}{1+\epsilon} 
              D\left(\Delta^{-}_{A\to B}(\varphi_{RA})\big\|\mE_{A\to B}^2(\varphi_{RA})\right).
\ee
On the other hand,
\begin{align}
  \label{2222}
  D&\left(\Omega_{A\to B}(\varphi_{RA})\big\|\mE_{A\to B}(\varphi_{RA})\right)=-H\left(\Omega_{A\to B}(\varphi_{RA})\right)-\tr\left(\Omega_{A\to B}(\varphi_{RA})\big\|\mE_{A\to B}(\varphi_{RA})\right)\nonumber\\
  &\geq -h\left(\frac{\epsilon}{1+\epsilon}\right)-\frac{1}{1+\epsilon}H\left(\mM_{A\to B}(\varphi_{RA})\right)
-\frac{\epsilon}{1+\epsilon} H\left(\Delta^{+}_{A\to B}(\varphi_{RA})\right)\nonumber\\
&\quad-\frac{1}{1+\epsilon}\tr\left(\mM_{A\to B}(\varphi_{RA})\log\mE_{A\to B}(\varphi_{RA})\right)
-\frac{\epsilon}{1+\epsilon}\tr\left(\Delta^{+}_{A\to B}(\varphi_{RA})\log\mE_{A\to B}(\varphi_{RA})\right)\nonumber\\
&=-h\left(\frac{\epsilon}{1+\epsilon}\right)+\frac{1}{1+\epsilon}D\left(\mM_{A\to B}(\varphi_{RA})\big\|\mE_{A\to B}(\varphi_{RA})\right)+\frac{\epsilon}{1+\epsilon} D\left(\Delta^{+}_{A\to B}(\varphi_{RA})\big\|\mE_{A\to B}(\varphi_{RA})\right).
\end{align}
Combining both~\eqref{1111} and~\eqref{2222} gives
\ba
\frac{1}{1+\epsilon}&\Big(D\left(\mM_{A\to B}(\varphi_{RA})\big\|\mE_{A\to B}(\varphi_{RA})\right)-D\left(\mN_{A\to B}(\varphi_{RA})\big\|\mE_{A\to B}^{1}(\varphi_{RA})\right)\Big)\\
&\leq h\left(\frac{\epsilon}{1+\epsilon}\right)+\frac{\epsilon}{1+\epsilon}\Big(D\left(\Delta^{-}_{A\to B}(\varphi_{RA})\big\|\mE_{A\to B}^2(\varphi_{RA})\right)-D\left(\Delta^{+}_{A\to B}(\varphi_{RA})\big\|\mE_{A\to B}(\varphi_{RA})\right)\Big).
\ea
In particular, 
\be
D\left(\mM_{A\to B}(\varphi_{RA})\big\|\mE_{A\to B}(\varphi_{RA})\right)-D\left(\mN_{A\to B}\big\|\mE_{A\to B}^{1}\right)
\leq (1+\epsilon)h\left(\frac{\epsilon}{1+\epsilon}\right)+\epsilon D\left(\Delta^{-}_{A\to B}\big\|\mE_{A\to B}^2\right).
\ee
Finally, choosing $\varphi_{RA}$, $\mE^1$, and $\mE^2$, such that 
\ba
&D\left(\mN_{A\to B}\big\|\mE_{A\to B}^{1}\right)=D_{\mf}\left(\mN_{A\to B}\right)\\
&D\left(\Delta^{-}_{A\to B}\big\|\mE_{A\to B}^{2}\right)=D_{\mf}\left(\Delta^{-}_{A\to B}\right)\\
&D\left(\mM_{A\to B}(\varphi_{RA})\big\|\mE_{A\to B}(\varphi_{RA})\right)=D\left(\mM_{A\to B}\big\|\mE_{A\to B}\right)\geq D_\mf(\mM_{A\to B}),
\ea
we conclude that
\ba
D_{\mf}\left(\mM_{A\to B}\right)-D_\mf\left(\mN_{A\to B}\right)
\leq (1+\epsilon)h\left(\frac{\epsilon}{1+\epsilon}\right)+\epsilon D_\mf\left(\Delta^{-}_{A\to B}\right)
\leq (1+\epsilon)h\left(\frac{\epsilon}{1+\epsilon}\right)+\epsilon \kappa.
\ea
This completes the proof.
\end{proof}

The proof above can be adjusted in order to prove the asymptotic continuity of $E_\mf$. However, it will be very useful to prove a slightly stronger version of the asymptotic continuity that incorporate both $D_\mf$ and $E_\mf$ as special cases. We will use this version in the subsequent sections.

\subsubsection{Stronger version}
Let $\mS(RA)$ be a set of density matrices in $\mD(RA)$. For any $\mN\in\cptp(A\to B)$ denote
\be
E_{\mf,\mS}(\mN)\eqdef \min_{\mM\in\mf(A\to B)}\max_{\rho\in\mS(RA)}D\left(\mN_{A\to B}(\rho_{RA})\|\mM_{A\to B}(\rho_{RA})\right)\;.
\ee
We will assume here that the extreme points of $\mS(RA)$ are pure states, so that w.l.o.g. $|R|=|A|$ and there is no need to take supremum over $|R|$.

\begin{lemma}\label{ac}
Let $\mf$ be a convex resource theory admitting the tensor product structure. Suppose also that for any system $B$, $\mf(B)$ contains a full rank state. For a fixed dimension $|R|$, let $\mS(RA)$ be a set of density matrices in $\mD(RA)$, whose extreme points are pure states. Further, let $\mN\in\cptp(A\to B)$, and let $\{\mM^\varphi\}_{\varphi\in\mS(RA)}$ be a set of CP maps (not necessarily channels) with the  property
\be\label{diam}
\left\|\mN_{A\to B}(\varphi_{RA})-\mM_{A\to B}^{\varphi}(\varphi_{RA})\right\|_1\leq\epsilon\quad\forall\;\varphi\in\mS(RA).
\ee
Then,
\begin{align}
  \label{asy}
  E_{\mf,\mS}(\mN_{A\to B}) 
    &-\max_{\varphi\in\mS(RA)}\min_{\mE\in\mf(A\to B)}
      D\left(\mM_{A\to B}^{\varphi}(\varphi_{RA})\|\mE_{A\to B}(\varphi_{RA})\right)\nonumber\\
    &\leq f(\epsilon)\log |AB|+\tr\left[\left(\gamma_{R}
             -\tr_B\big[\mM_{A\to B}^{\gamma}(\gamma_{RA})\big]\right)\log\gamma_R^{-1}\right],
\end{align}
where $f(\epsilon)$ is independent on the dimensions and satisfies 
$\lim_{\epsilon\to 0}f(\epsilon)=0$, and $\gamma_{RA}\in\mS(RA)$ is a pure 
state defined below in~\eqref{chi}.
\end{lemma}

\begin{remark}
For the case that for all $\varphi\in\mS(RA)$, $\mM^\varphi=\mM\in\cptp(A\to B)$ 
is CPTP and $\mS(RA)=\mD(RA)$ with $|R|=|A|$, Eq.~\eqref{diam} reduces to 
$\|\mM-\mN\|_\diamond\leq\epsilon$, and since $\mM$ is trace preserving, 
Eq.~\eqref{asy} reduces to
\be
  D_{\mf}(\mN_{A\to B})-D_{\mf}(\mM_{A\to B})\leq f(\epsilon)\log |AB|\;.
\ee
That is, we reproduce that $D_{\mf}(\mN_{A\to B})$ is asymptotically continuous. 
\end{remark}
\begin{remark}
For the case that for all $\varphi\in\mS(RA)$, $\mM^\varphi=\mM\in\cptp(A\to B)$ is CPTP and $\mS(RA)=\mf(RA)$, the lemma above gives
 \be
E_{\mf}(\mN_{A\to B})-E_{\mf}(\mM_{A\to B})\leq f(\epsilon)\log |AB|\;.
\ee
That is, $E_{\mf}$ is also asymptotically continuous.
\end{remark}
\begin{remark}
Since the trace norm is contractive under partial trace, from~\eqref{diam} it follows that 
\be
\left\|\gamma_{R}-\tr_B\big[\mM_{A\to B}(\gamma_{RA})\big]\right\|_1\leq\epsilon\;.
\ee
Therefore, we have the bound
\be
\tr\left[\left(\gamma_{R}-\tr_B\big[\mM_{A\to B}(\gamma_{RA})\big]\right)\log\gamma_R^{-1}\right]\leq \epsilon\log\|\gamma_R^{-1}\|_{\infty}\;.
\ee
\end{remark}

\begin{proof}
Denote by $J_{RB}^{\mN}\eqdef\mN_{A\to B}(\phi^{+}_{RA})$ the Choi matrix of 
a quantum channel $\mN\in\cptp(A\to B)$, and by
\be\label{kappa}
  \kappa\eqdef\min_{\omega\in\mf(B)}\log\|\omega_B^{-1}\|_{\infty}\;.
\ee 
Furthermore, for any $\varphi\in\mS(RA)$ denote by 
\be
\tau^{\pm,\varphi}_{RB}=\big(\mM^\varphi_{A\to B}\left(\varphi_{RA}\right)-\mN_{A\to B}\left(\varphi_{RA}\right)\big)_{\pm}
\ee
and observe that 
$\tr\left[\tau^{+,\varphi}_{RB}+\tau^{-,\varphi}_{RB}\right]\leq\epsilon$.
By definition, $\mM^\varphi_{A\to B}\left(\varphi_{RA}\right)-\mN_{A\to B}\left(\varphi_{RA}\right)=\tau^{+,\varphi}_{RB}-\tau^{-,\varphi}_{RB}$ so that
\be
\omega_{RB}\eqdef\frac{1}{1+\epsilon}\mM^\varphi_{A\to B}\left(\varphi_{RA}\right)+\frac{\epsilon}{1+\epsilon}\left(\frac{1}{\epsilon}\tau^{-,\varphi}_{RB}\right)=\frac{1}{1+\epsilon}\mN_{A\to B}(\varphi_{RA})+\frac{\epsilon}{1+\epsilon}\left(\frac{1}{\epsilon}\tau^{+,\varphi}_{RB}\right).
\ee
Also, define
$$\mE_{A\to B}\eqdef\frac{1}{1+\epsilon}\mE^{1}_{A\to B}+\frac{\epsilon}{1+\epsilon}\mE^2_{A\to B}\;,$$ where $\mE^1_{A\to B},\mE^2_{A\to B}\in\mf(A\to B)$ are free quantum channels.
With these definitions, for any channels as above we have from the joint convexity of the relative entropy
\ba
  \label{1}
  D&\left(\omega_{RB}\big\|\mE_{A\to B}(\varphi_{RA})\right)
    \leq \frac{1}{1+\epsilon}D\left(\mM^\varphi_{A\to B}\left(\varphi_{RA}\right)\big\|\mE_{A\to B}^1(\varphi_{RA})\right)+\frac{\epsilon}{1+\epsilon} D\left(\frac{1}{\epsilon}\tau^{-,\varphi}_{RB}\big\|\mE_{A\to B}^2(\varphi_{RA})\right)\\
   &\phantom{====}
    =\frac{1}{1+\epsilon}D\left(\mM^\varphi_{A\to B}\left(\varphi_{RA}\right)\big\|\mE_{A\to B}^1(\varphi_{RA})\right)-\frac{1}{1+\epsilon} \tr\left[\tau^{-,\varphi}_{RB}\log\mE_{A\to B}^2(\varphi_{RA})\right]-\frac{\epsilon}{1+\epsilon} H\left(\frac{1}{\epsilon}\tau^{-,\varphi}_{RB}\right).
\ea
On the other hand,
\begin{align}\label{2}
D&\left(\omega_{RB}\big\|\mE_{A\to B}(\varphi_{RA})\right)=-H\left(\omega_{RB}\right)-\tr\left(\omega_{RB}\big\|\mE_{A\to B}(\varphi_{RA})\right)\nonumber\\
&\geq -h\left(\frac{\epsilon}{1+\epsilon}\right)-\frac{1}{1+\epsilon}H\left(\mN_{A\to B}(\varphi_{RA})\right)-\frac{\epsilon}{1+\epsilon} H\left(\frac{1}{\epsilon}\tau^{+,\varphi}_{RB}\right)\nonumber\\
&\quad-\frac{1}{1+\epsilon}\tr\left(\mN_{A\to B}(\varphi_{RA})\log\mE_{A\to B}(\varphi_{RA})\right)
-\frac{1}{1+\epsilon}\tr\left(\tau^{+,\varphi}_{RB}\log\mE_{A\to B}(\varphi_{RA})\right)\nonumber\\
&=\frac{1}{1+\epsilon}D\left(\mN_{A\to B}(\varphi_{RA})\big\|\mE_{A\to B}(\varphi_{RA})\right)-\frac{1}{1+\epsilon}\tr\left(\tau^{+,\varphi}_{RB}\log\mE_{A\to B}(\varphi_{RA})\right)-\frac{\epsilon}{1+\epsilon} H\left(\frac{1}{\epsilon}\tau^{+,\varphi}_{RB}\right)-h\left(\frac{\epsilon}{1+\epsilon}\right).
\end{align}
Combining both~\eqref{1} and~\eqref{2} gives
\begin{align}
D&\left(\mN_{A\to B}(\varphi_{RA})\big\|\mE_{A\to B}(\varphi_{RA})\right)\\
&\leq D\left(\mM^\varphi_{A\to B}\left(\varphi_{RA}\right)\big\|\mE_{A\to B}^{1}(\varphi_{RA})\right)+\tr\left[\tau^{+,\varphi}_{RB}\log\mE_{A\to B}(\varphi_{RA})\right]-\tr\left(\tau^{-,\varphi}_{RB}\log\mE_{A\to B}^2(\varphi_{RA})\right)\\
&+(1+\epsilon)h\left(\frac{\epsilon}{1+\epsilon}\right)+\epsilon\left(H\left(\frac{1}{\epsilon}\tau^{+,\varphi}_{RB}\right)-H\left(\frac{1}{\epsilon}\tau^{-,\varphi}_{RB}\right)\right).
\end{align}
We now make a few observations. First, note that the last term in the 
equation above is bounded by
\be
\epsilon\left(H\left(\frac{1}{\epsilon}\tau^{+,\varphi}_{RB}\right)-H\left(\frac{1}{\epsilon}\tau^{-,\varphi}_{RB}\right)\right)\leq\epsilon\log|RB|\;.
\ee
Second, denote by $t\equiv 2^{D_{\max}(\mE^1\|\mE^2)}$ the smallest number 
satisfying $t\mE^2\geq\mE^1$, and observe that
\ba
\tr\left[\tau^{+,\varphi}_{RB}\log\mE_{A\to B}(\varphi_{RA})\right]&=\tr\left[\tau^{+,\varphi}_{RB}\log\left(\frac{1}{1+\epsilon}\mE^{1}_{A\to B}(\varphi_{RA})+\frac{\epsilon}{1+\epsilon}\mE^2_{A\to B}(\varphi_{RA})\right)\right]\\
&\leq \tr\left[\tau^{+,\varphi}_{RB}\log\left(\frac{t+\epsilon}{1+\epsilon}\mE^2_{A\to B}(\varphi_{RA})\right)\right]\\
&=\tr[\tau^{+,\varphi}_{RB}]\log\left(\frac{t+\epsilon}{1+\epsilon}\right)+\tr\left[\tau^{+,\varphi}_{RB}\log\mE_{A\to B}^2(\varphi_{RA})\right]\\
&\leq\epsilon\log t+\epsilon\log(1+\epsilon)+\tr\left[\tau^{+,\varphi}_{RB}\log\mE_{A\to B}^2(\varphi_{RA})\right],
\ea
where the first inequality follows from the operator monotonicity of the log function.
Therefore,
\ba\label{g1}
D\left(\mN_{A\to B}(\varphi_{RA})\big\|\mE_{A\to B}(\varphi_{RA})\right)&\leq D\left(\mM^{\varphi}_{A\to B}\left(\varphi_{RA}\right)\big\|\mE_{A\to B}^{1}(\varphi_{RA})\right)
+(1+\epsilon)h\left(\frac{\epsilon}{1+\epsilon}\right)+\epsilon\log\left(1+\epsilon\right)\\
&+\epsilon\log|AB|+\epsilon D_{\max}(\mE^1\|\mE^2)+\tr\left[(\tau^{+,\varphi}_{RB}-\tau^{-,\varphi}_{RB})\log\mE_{A\to B}^2(\varphi_{RA})\right].
\ea
Now, take $\mE^2_{A\to B}(X)\eqdef\tr[X]\omega_B$ to be a constant channel with the full rank state $\omega_B\in\mf(B)$ optimizing~\eqref{kappa}. Then,
\ba
D_{\max}(\mE^1\|\mE^2)&=\log\min\Big\{t\geq 0\;:\;tI^A\otimes\omega_B\geq J^{\mE^1}_{AB}\Big\}\\
&=\log \|\omega_B^{-1/2}J^{\mE^1}_{AB}\omega_B^{-1/2}\|_{\infty}\leq \log\|\omega_{B}^{-1}\|_{\infty}\|J^{\mE^1}_{AB}\|_{\infty}\leq\log|A|+\kappa.
\ea 
Hence, after minimizing both sides of~\eqref{g1} over $\mE^1\in\mf(A\to B)$ we get
\ba
\min_{\mE\in\mf(A\to B)}D\left(\mN_{A\to B}(\varphi_{RA})\big\|\mE_{A\to B}(\varphi_{RA})\right)&\leq \min_{\mE\in\mf(A\to B)}D\left(\mM^{\varphi}_{A\to B}\left(\varphi_{RA}\right)\big\|\mE_{A\to B}(\varphi_{RA})\right)
+(1+\epsilon)h\left(\frac{\epsilon}{1+\epsilon}\right)\\
&+\epsilon\log\left(1+\epsilon\right)+\epsilon\log|A|^2|B|+\epsilon \kappa+\tr\left[(\tau^{+,\varphi}_{RB}-\tau^{-,\varphi}_{RB})\log(\varphi_R\otimes\omega_B)\right].
\ea
Furthermore, let $\gamma_{RA}\in\mS(RA)$ be such that
\ba\label{chi}.
E_{\mf,\mS}\left(\mN_{A\to B}\right)=\min_{\mE\in\mf(A\to B)}D\left(\mN_{A\to B}(\gamma_{RA})\big\|\mE_{A\to B}(\gamma_{RA})\right)\;.
\ea
W.l.o.g. we can assume that $\gamma_{RA}$ is pure since the extreme points of 
$\mS(RA)$ are pure states. With this choice we have
\ba
E_{\mf,\mS}\left(\mN_{A\to B}\right)&\leq \max_{\varphi\in\mS(RA)}\min_{\mE\in\mf(A\to B)}D\left(\mM^{\varphi}_{A\to B}\left(\varphi_{RA}\right)\big\|\mE_{A\to B}(\varphi_{RA})\right)
+(1+\epsilon)h\left(\frac{\epsilon}{1+\epsilon}\right)\\
&\phantom{===}
+\epsilon\log\left(1+\epsilon\right)+\epsilon\log|A|^2|B|+\epsilon \kappa+\tr\left[(\tau^{+,\gamma}_{RB}-\tau^{-,\gamma}_{RB})\log(\gamma_R\otimes\omega_B)\right]
\label{lastterm}
\ea
What is therefore left is to bound the last term in the RHS of~\eqref{lastterm}. 
For pure $\gamma_{RA}$ (with $|R|=|A|$) set 
$|\gamma_{RA}\ra=\gamma_R^{1/2}\otimes I_{\tR}|\phi^{+}_{RA}\ra$, so that
\begin{align}
&\tr\left[(\tau^{+,\gamma}_{RB}-\tau^{-,\gamma}_{RB})\log(\gamma_{R}\otimes\omega_B)\right]\nonumber\\
&=\tr\left[\left(\mM^{\gamma}_{A\to B}(\gamma_{RA})-\mN_{A\to B}\left(\gamma_{RA}\right)\right)\log(\gamma_{R}\otimes\omega_B)\right]\nonumber\\
&=\tr\left[\big(\mM^{\gamma}_{A\to B}\left(\gamma_{A}\right)-\mN_{A\to B}\left(\gamma_{A}\right)\big)\log(\omega_B)\right]+\tr\left[\eta_R\log(\gamma_{R})\right],
\end{align}
where
\ba\label{38}
\eta_R&\eqdef\tr_B\left[\mM^{\gamma}_{A\to B}(\gamma_{RA})-\mN_{A\to B}\left(\gamma_{RA}\right)\right]\\
&=\tr_B\left[\mM^{\gamma}_{A\to B}(\gamma_{RA})\right]-\gamma_R\\
&=\gamma_R^{1/2}\left(J^{\mM^{\gamma}}_{R}-I_{R}\right)\gamma_R^{1/2}\;,
\ea
where $J^{\mM^{\gamma}}_{R}$ is the marginal of the Choi matrix of $\mM^{\gamma}$.
Further, using the fact that for any Hermitian operator $X$ we have $X\leq |X|\eqdef X_++X_-$ and $\tr|X|=\|X\|_1$,
\begin{align}
\tr\left[\big(\mM^{\gamma}_{A\to B}\left(\gamma_A\right)-\mN_{A\to B}\left(\gamma_{A}\right)\big)\log(\omega_B)\right]
&=\tr\left[\big(\mN_{A\to B}\left(\gamma_{A}\right)-\mM^{\gamma}_{A\to B}\left(\gamma_A\right)\big)\log(\omega_B^{-1})\right]\\
&\leq\tr\left[\big|\mN_{A\to B}\left(\gamma_{A}\right)-\mM^{\gamma}_{A\to B}\left(\gamma_A\right)\big|\log(\omega_B^{-1})\right]\\
&\leq\left\|\mN_{A\to B}\left(\gamma_{A}\right)-\mM^{\gamma}_{A\to B}\left(\gamma_A\right)\right\|_1\log\|\omega_{B}^{-1}\|_{\infty}\\
&\leq\epsilon \kappa.
\end{align}
Combining everything we get
\begin{align}
E_{\mf,\mS}\left(\mN_{A\to B}\right)
\leq \max_{\varphi\in\mS(RA)}\min_{\mE\in\mf(A\to B)}D\left(\mM_{A\to B}^{\varphi}\left(\varphi_{RA}\right)\big\|\mE_{A\to B}(\varphi_{RA})\right)+\tr\left[\gamma_R^{1/2}\left(I_{R}-J^{\mM^\gamma}_{R}\right)\gamma_R^{1/2}\log\gamma_R^{-1}\right]+g(\epsilon),
\end{align}
where
\be
g(\epsilon)\eqdef \epsilon\left(\log|A|^2|B|+2\kappa\right)+(1+\epsilon)h\left(\frac{\epsilon}{1+\epsilon}\right)+\epsilon\log\left(1+\epsilon\right).
\ee
This completes the proof.
\end{proof}

\section{The Asymptotic Equipartition Property (AEP)}

As defined in the main text the \emph{logarithmic robustness of a dynamical resource} 
$\mN\in\cptp(A\to B)$ is defined as
\ba
LR_\mf(\mN_{A\to B})&\eqdef\min_{\mE\in\mf(A\to B)}D_{\max}(\mN\|\mE)\\
&=\log_2\min\Big\{t\;:\;t\mE_{A\to B}\geq\mN_{A\to B}\;\;;\;\; \mE\in\mf(A\to B)\Big\},
\ea
where the notation $t\mE_{A\to B}\geq\mN_{A\to B}$ means that $t\mE_{A\to B}-\mN_{A\to B}$ is a CP map. We also define
\be
\underline{LR}_{\mf}(\mN_{A\to B})
  \eqdef \min_{\mE\in\mf(A\to B)}\sup_{\varphi\in\mf(RA)}
            D_{\max}\big(\mN_{A\to B}(\varphi_{RA})\|\mE_{A\to B}(\varphi_{RA})\big).
\ee
We will assume here that the extreme point of $\mf(RA)$ are pure states so that 
the optimization above over $\mf(RA)$ can be taken to be over pure states with $|R|=|A|$.

\subsubsection{Standard Smoothing}

The smoothed version of the logarithmic robustness can be defined as
\be
  \widetilde{LR}_\mf^{\epsilon}(\mN)\eqdef\min_{\mN'\in B_\epsilon(\mN)}LR_\mf(\mN'),
\ee
with the diamond-norm ball
\be 
  B_\epsilon(\mN)\eqdef\Big\{\mN'\in\cptp(A\to B)\;:\;\|\mN'-\mN\|_{\diamond}\leq \epsilon\Big\}.
\ee
The above smoothing of $LR_\mf$ is a straightforward generalization from states 
to channels. While we will adopt a different type of smoothing later on, we 
start by showing that the regularization of $\widetilde{LR}_\mf^{\epsilon}$ 
provides an upper bound on the regularization of $D_\mf$.

\begin{lemma}
Let $\mf$ be a convex QRT, and define
\ba
  \label{regularization}
  D_\mf^{\infty}(\mN)\eqdef\liminf_{n\to\infty}\frac{1}{n}D_\mf(\mN^{\otimes n})
    \quad;\quad 
  \widetilde{LR}_\mf^{\infty}(\mN)\eqdef\lim_{\epsilon\to 0}\liminf_{n\to\infty}\frac{1}{n}\widetilde{LR}_\mf^{\epsilon}(\mN^{\otimes n}).
\ea
Then,
\be
  \label{aepp}
  D_\mf^{\infty}(\mN)\leq \widetilde{LR}_\mf^{\infty}(\mN)\;.
\ee
\end{lemma}

\begin{proof}
Let $\mN^{\epsilon}_n\in\cptp(A^n\to B^n)$ and $\mE_n\in\mf(A^n\to B^n)$ be optimal channels such that $\|\mN^{\epsilon}_n-\mN^{\otimes n}\|_{\diamond}\leq\epsilon$ and  $LR_\mf^{\epsilon}(\mN^{\otimes n})=D_{\max}(\mN^{\epsilon}_n\|\mE_n)$. Using the fact that $D_{\max}$ is always greater that the relative entropy $D$, we conclude that
\be
\frac{1}{n}D_\mf(\mN^{\epsilon}_n)\leq\frac{1}{n}D(\mN^{\epsilon}_n\|\mE_n)\leq\frac{1}{n}D_{\max}(\mN^{\epsilon}_n\|\mE_n)=\frac{1}{n}\widetilde{LR}_\mf^{\epsilon}(\mN^{\otimes n}).
\ee
Now, since $D_\mf$ is asymptotically continuous there exists a function $f:\mbb{R}\to \mbb{R}$ with the property $\lim_{\epsilon\to 0}f(\epsilon)=0$ such that 
\be
\frac{1}{n}D_\mf(\mN^{\otimes n})\leq \frac{1}{n}D_\mf(\mN^{\epsilon}_n)+f(\epsilon).
\ee
Therefore, taking the limit $n\to \infty$ followed by $\epsilon\to 0$ on both sides gives
\be
\liminf_{n\to\infty}\frac{1}{n}D_\mf(\mN^{\otimes n})\leq \lim_{\epsilon\to 0}\liminf_{n\to\infty}\frac{1}{n}\widetilde{LR}_\mf^{\epsilon}(\mN^{\otimes n})\;.
\ee
This completes the proof.
\end{proof}

\subsubsection{Liberal Smoothing}

Let
\be 
B_\epsilon^{\varphi}(\mN)\eqdef\Big\{\mN'\in\text{CP}(A\to B)\;:\;\|\mN'_{A\to B}(\varphi_{RA})-\mN_{A\to B}(\varphi_{RA})\|_{1}\leq \epsilon\Big\},
\ee
and consider the following types of smoothing:
\begin{align}\label{smooth2}
  LR_{\mf}^{\epsilon}(\mN) 
    &\eqdef \max_{\varphi\in\mD(RA)}\min_{\mN'\in B_\epsilon^{\varphi}(\mN)}\min_{\mE\in\mf(A\to B)}
                 D_{\max}\big(\mN'_{A\to B}\|\mE_{A\to B}\big), \\
  \underline{LR}_{\mf}^{\epsilon}(\mN) 
    &\eqdef \sup_{\varphi\in\mf(RA)}\min_{\mN'\in B_\epsilon^{\varphi}(\mN)}\min_{\mE\in\mf(A\to B)}
                 D_{\max}\big(\mN'_{A\to B}(\varphi_{RA})\|\mE_{A\to B}(\varphi_{RA})\big).
\end{align}
Note that both of the above smoothings respect the condition that for $\epsilon=0$,
\ba
  LR_{\mf}^{\epsilon=0}(\mN)             &= LR_\mf(\mN), \\
  \underline{LR}_{\mf}^{\epsilon=0}(\mN) &= \underline{LR}_{\mf}(\mN_{A\to B}).
\ea

For each $\varphi\in\mD({RA})$, it holds $B_\epsilon(\mN)\subset B_\epsilon^\varphi(\mN)$, 
hence we have
\be
\widetilde{LR}_\mf^{\epsilon}(\mN)\geq \min_{\mN'\in B_\epsilon^\varphi(\mN)}LR_\mf(\mN')
\ee
Furthermore, since the above equation holds for all $\varphi\in\mD(RA)$ we must have
\be\label{ineq}
\widetilde{LR}_\mf^{\epsilon}(\mN)\geq LR_\mf^{\epsilon}(\mN)\;.
\ee
The above equation holds also even if we define $B_\epsilon^{\varphi}$ with respect to CPCP maps.
That is, define
\be 
{B'}_\epsilon^{\varphi}(\mN)\eqdef\Big\{\mN'\in\cptp(A\to B)\;:\;\|\mN'_{A\to B}(\varphi_{RA})-\mN_{A\to B}(\varphi_{RA})\|_{1}\leq \epsilon\Big\}
\ee
and
\be
{LR'}_{\mf}^{\epsilon}(\mN)\eqdef\max_{\varphi\in\mD(RA)}\min_{\mN'\in {B'}_\epsilon^{\varphi}(\mN)}\min_{\mE\in\mf(A\to B)}D_{\max}\big(\mN'_{A\to B}\|\mE_{A\to B}\big)\;.
\ee
Then, we also have
\be
\widetilde{LR}_\mf^{\epsilon}(\mN)\geq {LR'}_\mf^{\epsilon}(\mN)\;.
\ee
We now show that if the inequality above is strict, then also the inequality in~\eqref{aepp} 
is strict, and consequently the AEP cannot hold with standard smoothing.

\begin{lemma}
Let $\mf$ be a convex QRT, and define
$
{LR'}_\mf^{\infty}(\mN)\eqdef\lim_{\epsilon\to 0}\liminf_{n\to\infty}\frac{1}{n}{LR'}_\mf^{\epsilon}(\mN^{\otimes n})
$. Then,
\be
D_\mf^{\infty}(\mN)\leq {LR'}_\mf^{\infty}(\mN)\;.
\ee
\end{lemma}
\begin{proof}
For any $\epsilon>0$, $n\in\mbb{N}$, and $\varphi\in\mD(R^nA^n)$, let $\mM^{\varphi}_n\in\cptp(A^n\to B^n)$ and $\mE_n^{\varphi}\in\mf(A^n\to B^n)$ be optimal channels such that $\|\mM^{\varphi}_n(\varphi_{R^nA^n})-\mN^{\otimes n}(\varphi_{R^nA^n})\|_{1}\leq\epsilon$ and  
$$
\min_{\mN'\in {B'}_\epsilon^{\varphi}(\mN^{\otimes n})}\min_{\mE\in\mf(A^n\to B^n)}D_{\max}\big(\mN'\|\mE\big)=D_{\max}(\mM^{\varphi}_n\|\mE_n^{\varphi})\;.
$$ 
Using the fact that $D_{\max}$ is always greater that the relative entropy $D$, we conclude that
$$
\frac{1}{n}\min_{\mE\in\mf(A^n\to B^n)}D\left(\mM^{\varphi}_n(\varphi_{R^nA^n})\big|\mE(\varphi_{R^nA^n})\right)\leq\frac{1}{n}D(\mM^{\varphi}_n\|\mE_n^{\varphi})\leq\frac{1}{n}D_{\max}(\mM^{\varphi}_n\|\mE_n^{\varphi})=\frac{1}{n}\min_{\mN'\in {B'}_\epsilon^{\varphi}(\mN^{\otimes n})}\min_{\mE\in\mf(A^n\to B^n)}D_{\max}\big(\mN'\|\mE\big)
$$
Combining this with Lemma~\ref{ac} we have 
\be
\frac{1}{n}D_\mf(\mN^{\otimes n})\leq\frac{1}{n}\max_{\varphi\in\mD(RA)}\min_{\mE\in\mf(A^n\to B^n)}D\left(\mM^{\varphi}_n(\varphi_{R^nA^n})\big|\mE(\varphi_{R^nA^n})\right)+f(\epsilon)\leq \frac{1}{n}{LR'}_\mf^{\epsilon}(\mN^{\otimes n})+f(\epsilon)\;.
\ee
Therefore, taking the limit $n\to \infty$ followed by $\epsilon\to 0$ on both sides gives
\be
\liminf_{n\to\infty}\frac{1}{n}D_\mf(\mN^{\otimes n})\leq \lim_{\epsilon\to 0}\liminf_{n\to\infty}\frac{1}{n}{LR'}_\mf^{\epsilon}(\mN^{\otimes n})\;.
\ee
This completes the proof.
\end{proof}

The above lemma demonstrates  that if the standard smoothing leads to different quantities than the liberal smoothing then 
AEP cannot hold when the quantities are defined with respect to the standard smoothing. This is the reason why we adopt here this new type of smoothing.

The liberal smoothing is strongly connected to the underlying QRT. In particular, the functions $LR_{\mf}^{\epsilon}(\mN)$ and $\underline{LR}_{\mf}^{\epsilon}(\mN)$ remain resource monotones (see the lemma below). Note also that by definition
\be
\underline{LR}_{\mf}^{\epsilon}(\mN_{A\to B})\leq LR_\mf^{\epsilon}(\mN_{A\to B})\;.
\ee
\begin{lemma}
Let $\Theta:\cptp(A\to B)\to\cptp(A'\to B')$ be a superchannel defined by
\be
\Theta[\mN_{A\to B}]\eqdef \mE^{\post}_{BE\to B'}\circ\mN_{A\to B}\circ\mE^{\pree}_{A'\to AE},
\ee
with $\mE^{\pree}\in\cptp(A'\to AE)$ and $\mE^{\post}\in\cptp(BE\to B')$ being completely RNG.
Then, 
\be
\underline{LR}_{\mf}^{\epsilon}\left(\Theta[\mN_{A\to B}]\right)\leq \underline{LR}_{\mf}^{\epsilon}\left(\mN_{A\to B}\right)
\quad;\quad
{LR}_{\mf}^{\epsilon}\left(\Theta[\mN_{A\to B}]\right)\leq {LR}_{\mf}^{\epsilon}\left(\mN_{A\to B}\right).
\ee
\end{lemma}
\begin{proof}
For any channel $\mN\in\cptp(A\to B)$, we have
\ba
\underline{LR}_{\mf}^{\epsilon}\left(\Theta[\mN_{A\to B}]\right)&=\sup_{\varphi\in\mf(R'A')}\min_{\mN'\in B_\epsilon^{\varphi}(\Theta[\mN])}\min_{\Phi\in\mf(A'\to B')}D_{\max}\big(\mN'_{A'\to B'}(\varphi_{R'A'})\|\Phi_{A'\to B'}(\varphi_{R'A'})\big)\\
&\leq \sup_{\varphi\in\mf(R'A')}\min_{\substack{\mM\in\text{CP}(A\to B)\\ \|\Theta[\mM-\mN](\varphi_{R'A'})\|_1\leq\epsilon}}\min_{\Phi\in\mf(A'\to B')}D_{\max}\big(\Theta[\mM](\varphi_{R'A'})\|\Phi_{A'\to B'}(\varphi_{R'A'})\big)\\
&\leq \sup_{\varphi\in\mf(R'A')}\min_{\substack{\mM\in\text{CP}(A\to B)\\ \|\Theta[\mM-\mN](\varphi_{R'A'})\|_1\leq\epsilon}}\min_{\Omega\in\mf(A\to B)}D_{\max}\big(\Theta[\mM](\varphi_{R'A'})\|\Theta[\Omega](\varphi_{R'A'})\big)\\
&\leq \sup_{\varphi\in\mf(R'A')}\min_{\substack{\mM\in\text{CP}(A\to B)\\ \|\Theta[\mM-\mN](\varphi_{R'A'})\|_1\leq\epsilon}}\min_{\Omega\in\mf(A\to B)}D_{\max}\big(\mM_{A\to B}\circ\mE^{\pree}_{A'\to AE}(\varphi_{R'A'})\|\Omega_{A\to B}\circ\mE^{\pree}_{A'\to AE}(\varphi_{R'A'})\big)\\
&\leq \sup_{\sigma\in\mf(R'AE)}\min_{\substack{\mM\in\text{CP}(A\to B)\\ \|\mE^{\post}_{BE\to B'}\circ(\mM-\mN)(\sigma_{R'AE})\|_1\leq\epsilon}}\min_{\Omega\in\mf(A\to B)}D_{\max}\big(\mM_{A\to B}(\sigma_{R'AE})\|\Omega_{A\to B}(\sigma_{R'AE})\big)\\
&\leq \sup_{\sigma\in\mf(R'AE)}\min_{\substack{\mM\in\text{CP}(A\to B)\\ \|(\mM-\mN)(\sigma_{R'AE})\|_1\leq\epsilon}}\min_{\Omega\in\mf(A\to B)}D_{\max}\big(\mM_{A\to B}(\sigma_{R'AE})\|\Omega_{A\to B}(\sigma_{R'AE})\big)\\
&= \sup_{\sigma\in\mf(R'AE)}\min_{\mM\in B^{\sigma}_{\epsilon}(\mN)}\min_{\Omega\in\mf(A\to B)}D_{\max}\big(\mM_{A\to B}(\sigma_{R'AE})\|\Omega_{A\to B}(\sigma_{R'AE})\big)\\
&=\underline{LR}_{\mf}^{\epsilon}\left(\mN_{A\to B}\right)\;.
\ea
The second line follows by restricting $\mN'$ to have the form $\Theta[\mM]$. The third line by restricting $\Phi$ to have the form $\Theta[\Omega]$. The fourth line from data processing inequality of $D_{\max}$. The fifth line by substituting 
$\sigma_{R'AE}=\mE^{\pree}_{A'\to AE}(\varphi_{R'A'})$ and then optimizing over all $\sigma\in\mf(R'AE)$. The sixth line from the contractivity of the trace norm, and finally, the seventh and eighth by definition.
The monotonicity of ${LR}_{\mf}^{\epsilon}$ follows similar lines.
\end{proof}

\subsection{Product-State Regularization}

One can define the regularized version of $D_\mf$ and $LR_\mf^{\epsilon}$ as 
in~\eqref{regularization}.
Note, however, that unlike the analogous quantity in the state domain, for 
channels the limit $n\to\infty$ of $\frac{1}{n}D_{\mf}(\mN^{\otimes n})$ 
may not exist in general, so we had to take in~\eqref{regularization} the 
$\liminf$ instead.
Moreover, it could even be that for some $\mN$
\be
  D_{\mf}^{\infty}(\mN)>D_{\mf}(\mN)\quad\text{and even}\quad D_{\mf}^{\infty}(\mN^{\otimes 2})>2D_{\mf}^{\infty}(\mN)\;!
\ee
Therefore, this type of regularization does not seem to be very promising, and we will adopt a different type of regularization that avoid these complications.

The type of regularization that we consider here is as follows. For each $n\in\mbb{N}$, and a channel $\mN\in\cptp(A\to B)$, we define the quantities
\ba
D_{\mf}^{(n)}(\mN) &\eqdef \frac{1}{n}\max_{\varphi\in\mD(RA)}\min_{\mE\in\mf(A^n\to B^n)}D\left(\mN_{A\to B}^{\otimes n}(\varphi_{RA}^{\otimes n})\|\mE_{A^n\to B^n}(\varphi_{RA}^{\otimes n})\right),\\
E_{\mf}^{(n)}(\mN) &\eqdef \frac{1}{n}\max_{\varphi\in\mf(RA)}\min_{\mE\in\mf(A^n\to B^n)}D\left(\mN_{A\to B}^{\otimes n}(\varphi_{RA}^{\otimes n})\|\mE_{A^n\to B^n}(\varphi_{RA}^{\otimes n})\right).
\ea
To motivate these definition, we first discuss some of their properties.

First, note that if $\mN\in\cptp(A\to B)$ is the constant channel $\mN_{A\to B}(X_A)=\tr[X_A]\omega_B$ then
\be
D_{\mf}^{(n)}(\mN)=E_{\mf}^{(n)}(\mN)=\frac{1}{n}D_{\mf}(\omega_B^{\otimes n})\eqdef\frac{1}{n}\min_{\sigma\in\mf(B^n)}D(\omega_B^{\otimes n}\|\sigma_{B^n})\;.
\ee
since both $D_\mf(\mN)$ and $E_\mf(\mN)$ reduces to $D_\mf(\omega_B)$ for replacement channels. Therefore, this type of regularization, reduces to the standard one when $\mN$ is a replacement channel. Next, we prove the following lemma.

\begin{lemma}
For any $\mN\in\cptp({A\to B})$ we have
\be
D_{\mf}^{(n+m)}(\mN_{A\to B})\leq \frac{n}{n+m}D_{\mf}^{(n)}(\mN_{A\to B})+\frac{m}{n+m}D_{\mf}^{(m)}(\mN_{A\to B}). 
\ee
The same relation also holds for $E_\mf^{(n)}$.
\end{lemma}
\begin{proof}
We have
\ba
(n+m)E_{\mf}^{(n+m)}(\mN_{A\to B})&=\sup_{\varphi\in\mf(RA)}\min_{\mE\in\mf(A^{n+m}\to B^{n+m})}D\left(\mN_{A\to B}^{\otimes {(n+m)}}\left(\varphi_{RA}^{\otimes {(n+m)}}\right)\big\|\mE\left(\varphi_{RA}^{\otimes {(n+m)}}\right)\right)\\
&\leq \sup_{\varphi\in\mf(RA)}\min_{\substack{\mE^1\in\mf(A^{n}\to B^{n})\\ \mE^2\in\mf(A^{m}\to B^{m})}}D\left(\big(\mN_{A\to B}(\varphi_{RA})\big)^{\otimes {(n+m)}}\big\|\mE^1\left(\varphi_{RA}^{\otimes {n}}\right)\otimes \mE^2\left(\varphi_{RA}^{\otimes {m}}\right)\right)\\
&=nE_{\mf}^{(n)}(\mN_{A\to B})+mE_{\mf}^{(m)}(\mN_{A\to B})\;.
\ea
The same lines of reasoning holds for $D_\mf$ as well.
\end{proof}

This lemma implies that the limits of $\mD_\mf^{(n)}$ and $E_{\mf}^{(n)}$, as $n\to\infty$, exist.
We therefore define the regularized version of $D_{\mf}$ and $E_{\mf}$ to be
\be
D_{\mf}^{(\infty)}(\mN)=\lim_{n\to\infty}E_{\mf}^{(n)}(\mN)\quad\text{and}\quad E_{\mf}^{(\infty)}(\mN)=\lim_{n\to\infty}E_{\mf}^{(n)}(\mN)\;.
\ee
From the lemmas above, the regularized quantities above satisfy
\be
D_{\mf}^{(\infty)}(\mN)\leq D_{\mf}^{(n)}(\mN)\quad\text{and}\quad E_{\mf}^{(\infty)}(\mN)\leq E_{\mf}^{(n)}(\mN)\quad\forall\;n\in\mbb{N}\;\;,\;\;\forall\;\mN\in\cptp(A\to B)
\ee 
and they are also resource monotones. Furthermore, note that the product-state regularization, $D_{\mf}^{(\infty)}(\mN)$, is no greater than the standard regularization $D_{\mf}^{\infty}(\mN)$ as defined in~\eqref{regularization}.
 
We can use this regularization method also for the smoothed logarithmic robustness 
quantities $LR_{\mf}^{\epsilon}$ and $\underline{LR}_{\mf}^{\epsilon}$. Define
\begin{align}
  LR^{\epsilon,n}_{\mf}(\mN) &\eqdef\frac{1}{n}\sup_{\varphi\in\mD(RA)}\min_{\mN'\in B_\epsilon^{\varphi^{\otimes n}}(\mN^{\otimes n})}\min_{\mE\in\mf(A^n\to B^n)}D_{\max}\big(\mN'_{A^n\to B^n}\|\mE_{A^n\to B^n}\big), \\
  \underline{LR}^{\epsilon,n}_{\mf}(\mN) &\eqdef\frac{1}{n}\sup_{\varphi^{\otimes n}\in\mF_n(RA)}\min_{\mN'\in B_\epsilon^{\varphi^{\otimes n}}(\mN^{\otimes n})}\min_{\mE\in\mf(A^n\to B^n)}D_{\max}\big(\mN'_{A^n\to B^n}(\varphi^{\otimes n})\|\mE_{A^n\to B^n}(\varphi^{\otimes n})\big), \\
  LR^{(\infty)}_{\mf}(\mN) &\eqdef\lim_{\epsilon\to 0}\liminf_{n\to\infty}LR^{\epsilon,n}_{\mf}(\mN), \\
  \underline{LR}^{(\infty)}_{\mf}(\mN) &\eqdef\lim_{\epsilon\to 0}\liminf_{n\to\infty}\underline{LR}^{\epsilon,n}_{\mf}(\mN).
\end{align}

\subsection{Proof of Theorem~\ref{aep}}

\begin{theorem*}
For all $\mN\in\cptp(A\to B)$,
\ba
  D_{\mf}^{(\infty)}(\mN)
    &= \lim_{\epsilon\to 0}\limsup_{n\to\infty}\frac{1}{n}{LR}_{\mf}^{\epsilon,n}(\mN)
     = \lim_{\epsilon\to 0}\liminf_{n\to\infty}\frac{1}{n}{LR}_{\mf}^{\epsilon,n}(\mN)
     \equiv LR^{(\infty)}_{\mf}(\mN), \\
  {E}_{\mf}^{(\infty)}(\mN)
    &= \lim_{\epsilon\to 0}\limsup_{n\to\infty}\frac{1}{n}\underline{LR}_{\mf}^{\epsilon,n}(\mN)
     =\lim_{\epsilon\to 0}\liminf_{n\to\infty}\frac{1}{n}\underline{LR}_{\mf}^{\epsilon,n}(\mN)
     \equiv \underline{LR}^{(\infty)}_{\mf}(\mN).
\ea
\end{theorem*}

\begin{proof}
We prove the theorem in two steps. First we prove the inequality
\be
E_{\mf}^{(\infty)}(\mN)\leq \underline{LR}^{(\infty)}_{\mf}(\mN)\quad\forall\mN\in\cptp(A\to B)\;.
\ee
Let $\epsilon>0$ and $\varphi\in\mf(RA)$. Let $\mM^{\varphi}_n\in\text{CP}(A^n\to B^n)$ 
be the optimal CP map such that 
\be
\|\mM^{\varphi}_n(\varphi^{\otimes n})-\mN^{\otimes n}(\varphi^{\otimes n})\|_{1}\leq\epsilon
\ee
and
\be
\min_{\substack{\mN'\in B_\epsilon^{\varphi^{\otimes n}}(\mN^{\otimes n})\\ \mE\in\mf(A^n\to B^n)}}D_{\max}\big(\mN'(\varphi^{\otimes n})\|\mE(\varphi^{\otimes n})\big)=\min_{\mE\in\mf(A^n\to B^n)}D_{\max}\big(\mM^{\varphi}_n(\varphi^{\otimes n})\|\mE(\varphi^{\otimes n})\big)\;.
\ee
Since $D_{\max}(\rho\|\sigma)\geq D(\rho\|\sigma)$ for all $\rho$ and $\sigma$ 
it follows from the above equation that
\be
\frac{1}{n}\min_{\mE\in\mf(A^n\to B^n)}D(\mM^{\varphi}_n(\varphi^{\otimes n})\|\mE(\varphi^{\otimes n}))\leq\frac{1}{n}\min_{\substack{\mN'\in B_\epsilon^{\varphi^{\otimes n}}(\mN^{\otimes n})\\ \mE\in\mf(A^n\to B^n)}}D_{\max}\big(\mN'(\varphi^{\otimes n})\|\mE(\varphi^{\otimes n})\big).
\ee
Therefore, taking the maximum over $\varphi\in\mf(RA)$ on both sides gives
\ba
\frac{1}{n}\max_{\varphi\in\mf(RA)}\min_{\mE_n\in\mf(A^n\to B^n)}D(\mM^{\varphi}_n(\varphi^{\otimes n})\|\mE_n(\varphi^{\otimes n}))\leq\frac{1}{n}\underline{LR}_{\mf}^{\epsilon,n}(\mN^{\otimes n}).
\ea
Combining this with the asymptotic continuity (see Lemma~\ref{ac} with $\mS$ being 
the set whose extreme points are the states of the form $\varphi^{\otimes n}$ with 
$\varphi\in\mf(A\to B)$) gives
 \ba\label{xgx}
 \frac{1}{n}E_{\mf}^{(n)}(\mN)\leq\frac{1}{n}\underline{LR}_{\mf}^{\epsilon,n}(\mN)+f(\epsilon)\log|AB|+\frac{1}{n}\tr\left[\left(\gamma_{R}^{\otimes n}-\tr_B\big[\mM^{\gamma}_{n}\left(\gamma_{RA}^{\otimes n}\right)\big]\right)\log\left(\gamma_R^{-1}\right)^{\otimes n}\right],
\ea
where  $\gamma_{RA}$ is defined  such that 
\ba\label{gamma}
E_{\mf}^{(n)}(\mN)&=\max_{\varphi\in\mf(RA)}\min_{\mE_n\in\mf(A^n\to B^n)}D\left(\mN^{\otimes n}(\varphi^{\otimes n})\big\|\mE_{n}(\varphi^{\otimes n})\right)\\
&=\min_{\mE_n\in\mf(A^n\to B^n)}D\left(\mN^{\otimes n}(\gamma_{RA}^{\otimes n})\big\|\mE_{n}(\gamma_{RA}^{\otimes n})\right)\;.
\ea
All that is left to show is that the last term in~\eqref{xgx} goes to zero.
Note that  $\gamma_{RA}$ can depend on $n$. Therefore, we will use the notation 
$\omega_n\equiv\gamma_R\in\mf(R)$ to emphasize this dependence.  
 
Let $\{k\}$ be a subsequence such that
\be
 \lim_{k\to\infty}\frac{1}{k}{LR}_{\mf}^{\epsilon,k}(\mN^{\otimes {k}})
   = \liminf_{n\to\infty}\frac{1}{n}{LR}_{\mf}^{\epsilon,n}(\mN^{\otimes n}).
\ee
To simplify the notations, we used the notation $k$ instead of something like $n_k$. 
Now recall that 
$$\|\mM^{\gamma}_n\left(\gamma_{RA}^{\otimes n}\right)-\mN^{\otimes n}\left(\gamma_{RA}^{\otimes n}\right)\|_{1}\leq\epsilon\;,$$ 
and in particular, from the contractivity of the trace norm,
\be\label{x}
 \|\tr_{B^n}\left[\mM^{\gamma}_n\left(\gamma_{RA}^{\otimes n}\right)\right]-\gamma_{R}^{\otimes n}\|_{1}\leq\epsilon\quad\forall\;n\in\mbb{N}.
\ee
Therefore, if $\|(\omega_{k}^{-1})^{\otimes k}\|_{\infty}$ is bounded, then
\be
 \frac{1}{k}\tr\left[\left(\omega_{k}^{\otimes k}-\tr_B\big[\mM^{\gamma}_{k}\left(\gamma_{RA}^{\otimes k}\right)\big]\right)\log\left(\omega_{k}^{-1}\right)^{\otimes k}\right]\leq\epsilon\log\|\omega_{k}^{-1}\|_{\infty}
\ee
is bounded and goes to zero as $\epsilon\to 0$. We therefore assume now that $\|\omega_{k}^{-1}\|_{\infty}$ is not bounded. Then, there exists a subsequence $\{j\}\subset\{k\}$ such that $\lambda_{\min}(\omega_{j})\to 0$ as $j\to\infty$. Next, we continue to check if there exists a subsequence of $\{\omega_{j}\}$ for which the second smallest eigenvalue of $\omega_{j}$ also goes to zero. If there isn't then we stop. Otherwise, we continue in this way until we find a subsequence of $n$, lets call it again for simplicity $\{k\}$, such that the first $m$ largest eigenvalues of $\omega_{k}$ are bounded from below, and the remaining $|R|-m$ eigenvalues are all going to zero in the limit $k\to\infty$.

We now bound the term 
\be
 \frac{1}{k}\tr\left[\left(\omega_{k}^{\otimes k}-\tr_B\big[\mM^{\gamma}_{k}\left(\gamma_{RA}^{\otimes k}\right)\big]\right)\log\left(\omega_{k}^{-1}\right)^{\otimes k}\right],
\ee
which can be expressed equivalently as
\be
 \frac{1}{k}\tr\left[(\omega_{k}^{1/2})^{\otimes k}\left(I_{R^{k}}-J^{\mM^{\gamma}_{k}}_{R^{k}}\right)\left(\omega_{k}^{1/2}\right)^{\otimes k}\log(\omega_{k}^{-1})^{\otimes k}\right],
\ee
where $J^{\mM^{\gamma}_{k}}_{R^{k}}$ is the marginal of the Choi matrix of $\mM^{\gamma}_{k}$.
Next, observe that
\be
 \log(\omega_{k}^{-1})^{\otimes k}=\big(\log\omega_{k}^{-1}\otimes I_R\otimes\cdots\otimes I_R\big)+\cdots+\big(I_R\otimes \cdots\otimes I_R\otimes \log\omega_{k}^{-1}\big)
\ee
It is therefore enough to bound each of the terms
\be
\tr\left[(\omega_{k}^{1/2})^{\otimes k}\left(I_{R^{k}}-J^{\mM^{\gamma}_{k}\circ\mE^{\gamma}_{k}}_{R^{k}}\right)\left(\omega_{k}^{1/2}\right)^{\otimes k}\left(\log\omega_{k}^{-1}\otimes I_R\otimes\cdots\otimes I_R\right)\right]
 = \tr\left[\omega_{k}^{1/2}\xi_R\omega_{k}^{1/2}\log\omega_{k}^{-1}\right],
\ee
where
\be
 \xi_R\equiv\tr_{\neq 1}\left[\left(I_R\otimes\left(\omega_{k}^{1/2}\right)^{\otimes (k-1)}\right)\left(I_{R^{k}}-J^{\mM^{\gamma}_{k}\circ\mE^{\gamma}_{k}}_{R^{k}}\right)\left(I_R\otimes\left(\omega_{k}^{1/2}\right)^{\otimes (k-1)}\right)\right],
\ee
with $\tr_{\neq 1}$ denoting a trace over all the $k$ $R$-systems except for the first one. 
Note that from~\eqref{x}  we have
$\tr[(\omega_{k}^{1/2}\xi_R\omega_{k}^{1/2})_+]\leq\epsilon$. Now, decompose $\omega_k=\alpha_k+\beta_k$, where $\alpha_k=\omega_kP_k$, and $P_k$ is the projection to the eigenspace of the $m$ largest eigenvalues of $\omega_k$, and $\beta_k=\omega_k(I_R-P_k)$. Since $\alpha_k\beta_k=\beta_k\alpha_k=0$ we have
\be
\tr\left[\omega_{k}^{1/2}\xi_R\omega_{k}^{1/2}\log\omega_{k}^{-1}\right]=\tr\left[\alpha_{k}^{1/2}\xi_R\alpha_{k}^{1/2}\log\alpha_{k}^{-1}\right]+\tr\left[\beta_{k}^{1/2}\xi_R\beta_{k}^{1/2}\log\beta_{k}^{-1}\right],
\ee
where the inverses of $\alpha_k$ and $\beta_k$ understood as the generalized inverses. 
Now, observe that
\be
\tr\left[\alpha_{k}^{1/2}\xi_R\alpha_{k}^{1/2}\log\alpha_{k}^{-1}\right]=\tr\left[\omega_{k}^{1/2}\xi_R\omega_{k}^{1/2}P_k\log\alpha_{k}^{-1}\right]\leq\tr[(\omega_{k}^{1/2}\xi_R\omega_{k}^{1/2})_+]\log\|\alpha_k^{-1}\|_\infty\leq\epsilon\log\|\alpha_k^{-1}\|_\infty,
\ee
where $\|\alpha_k^{-1}\|_\infty$ is bounded. For the other term, note that by definition, 
since $\mM^{\gamma}_{k}$ is a CP map, its Choi matrix is positive semidefinite so 
that $\xi_R\leq I_R$. Hence,
\be
\tr\left[\beta_{k}^{1/2}\xi_R\beta_{k}^{1/2}\log\beta_{k}^{-1}\right]=\tr\left[\xi_R\beta_{k}^{1/2}(\log\beta_{k}^{-1})\beta_{k}^{1/2}\right]
\leq \tr\left[\beta_{k}^{1/2}(\log\beta_{k}^{-1})\beta_{k}^{1/2}\right]\to 0
\ee
as $k\to\infty$ (since $\beta_k\to 0$ as $k\to\infty$). 
To summarize, there exists some constant $c>0$ such that for sufficiently large $k$
\be
 \tr\left[\omega_{k}^{1/2}\eta_R\omega_{k}^{1/2}\log\omega_{k}^{-1}\right]\leq \epsilon c\;.
\ee
Since this bound holds for each of the $k$ terms, we conclude that
\be
\frac{1}{k}\tr\left[\left(\omega_{k}^{1/2}\right)^{\otimes k}\left(I_{R^{k}}-J^{\mM^{\gamma}_{k}}_{R^{k}}\right)\left(\omega_{k}^{1/2}\right)^{\otimes k}\log\left(\omega_{k}^{-1}\right)^{\otimes k}\right].\leq\epsilon c
\ee
Therefore, by taking on both sides of~\eqref{xgx} the limit $n\to\infty$ followed by $\epsilon\to 0$ gives 
\be
 E_{\mf}^{(\infty)}(\mN)\leq\lim_{\epsilon\to 0}\liminf_{n\to\infty}\frac{1}{n}\underline{LR}_{\mf}^{\epsilon,n}(\mN^{\otimes n})\equiv \underline{LR}^{(\infty)}_{\mf}(\mN)\;.
\ee

We next prove the inequality
\be
E_{\mf}^{(\infty)}(\mN)\geq \lim_{\epsilon\to 0}\limsup_{n\to\infty}\frac{1}{n}\underline{LR}_{\mf}^{\epsilon,n}(\mN^{\otimes n})\quad\forall\mN\in\cptp(A\to B)\;.
\ee
This  inequality follows by a reasoning very similar to that given in~\cite{BP2010} 
for the sate domain. Let $\epsilon>0$ and define
\be
r_m\eqdef E_{\mf}^{(m)}(\mN)+\epsilon=\max_{\varphi\in\mf(RA)}\min_{\mE\in\mf(A^{m}\to B^{m})}D\left(\left(\mN_{A\to B}(\varphi_{RA})\right)^{\otimes m}\|\mE_{A^m\to B^m}(\varphi^{\otimes m}_{RA})\right)+\epsilon, 
\ee
We will also denote by $\mE^\varphi\in\mf({A^m\to B^m})$ the optimal 
channel in $\mf(A^m\to B^m)$ that satisfies
\be
\min_{\mE\in\mf(A^{m}\to B^{m})}D\left(\left(\mN_{A\to B}(\varphi_{RA})\right)^{\otimes m}\|\mE(\varphi^{\otimes m}_{RA})\right)=D\left(\left(\mN_{A\to B}(\varphi_{RA})\right)^{\otimes m}\|\mE^\varphi(\varphi^{\otimes m}_{RA})\right).
\ee
For every $n\in\mbb{N}$ and $\varphi_{RA}\in \mf(RA)$, we have
\be
\label{mn}
\left(\mN(\varphi_{RA})\right)^{\otimes mn}\leq 2^{nr_m}\left(\mE^\varphi(\varphi_{RA}^{\otimes m})\right)^{\otimes n}+\left(\left(\mN(\varphi_{RA})\right)^{\otimes mn}-2^{nr_m}\left(\mE^\varphi(\varphi_{RA}^{\otimes m})\right)^{\otimes n}\right)_{+}.
\ee
Denote by
\be\label{one}
\delta_{mn}\eqdef\tr\left(\left(\mN(\varphi_{RA})\right)^{\otimes mn}-2^{nr_m}\left(\mE^\varphi(\varphi_{RA}^{\otimes m})\right)^{\otimes n}\right)_{+}
\ee
From~\cite{Ogawa-2000} we have
\be\label{on}
\delta_{mn}\leq 2^{-n(r_m t-f(t))},
\ee
where $t\in[0,1]$ and 
\be
f(t)=\log\tr\left[\left(\mN^{\otimes m}(\varphi_{RA}^{\otimes m})\right)^{1+t}\left(\mE^\varphi(\varphi_{RA}^{\otimes m})\right)^{-t}\right].
\ee
Note that $f(0)=0$ and 
\be
f'(0)=D(\mN^{\otimes m}(\varphi_{RA}^{\otimes m})\|\mE^\varphi(\varphi_{RA}^{\otimes m}))\leq E_{\mf}^{(m)}(\mN)=r_m-\epsilon.
\ee
Hence, for small enough $t>0$ we get that $r_m t-f(t)>0$ which together with~\eqref{on} proves that $\lim_{n\to\infty}\delta_{mn}=0$ for all $m\in\mbb{N}$ and all $\varphi\in\mf(RA)$.
Now, recall the following lemma.
\begin{lemma*}[\cite{Datta-2009,BP2010}]
Let $\rho$ and $\sigma$ be two density matrices, and $P\geq 0$ be some positive semidefinite operator satisfying
$\rho\leq P+\epsilon\sigma$ for some $\epsilon>0$. Then, there exists a density matrix $\trho$ satisfying
\be
\trho\leq\frac{1}{1-\epsilon}P\quad\text{and}\quad\|\rho-\trho\|_1\leq 4\sqrt{\epsilon}\;.
\ee
\end{lemma*}
From this lemma and~\eqref{mn} it follows that there exists a sequence of 
density matrices $\eta_{R^{nm}B^{nm}}$ such that
\be
\|\left(\mN(\varphi_{RA})\right)^{\otimes mn} -\eta_{R^{mn}B^{mn}}\|_1\leq 4\sqrt{\delta_{mn}}\quad\text{and}\quad\eta_{R^{mn}B^{mn}}\leq \frac{1}{1-\delta_{mn}}2^{nr_m}\left(\mE^\varphi(\varphi_{RA}^{\otimes m})\right)^{\otimes n}.
\ee
Now, define the  the CP map $\mM_{mn}^{\varphi}\in\text{CP}({A^{mn}\to B^{mn}})$ that satisfy
\be
\eta_{R^{mn}B^{mn}}=\mM^\varphi_{{mn}}\left(\varphi_{RA}^{\otimes mn}\right).
\ee
Such a CP map always exists as long as the bipartite state $\varphi_{RA}$ is 
pure. This also implies that 
\be
\mM^\varphi_{{mn}}
\leq \frac{1}{1-\delta_{mn}}2^{nr_m}\left(\mE^\varphi\right)^{\otimes n}\quad\text{and}\quad \left\|\left(\mN_{A\to B}(\varphi_{RA})\right)^{\otimes mn} -\mM_{{mn}}^{\varphi}\left(\varphi_{RA}^{\otimes mn}\right)\right\|_1\leq 4\sqrt{\delta_{mn}}\;.
\ee
Let $n$ be large enough such that $4\sqrt{\delta_{nm}}\leq\epsilon$. Hence, 
\begin{align}
\underline{LR}_{\mf}^{\epsilon,mn}(\mN)&=\max_{\varphi\in\mf(RA)}\min_{\mN'\in B_\epsilon^{\varphi^{\otimes nm}}\left(\mN^{\otimes nm}\right)}\min_{\mE'\in\mf(A^{nm}\to B^{nm})}D_{\max}\big(\mN'\left(\varphi_{RA}^{\otimes mn}\right)\|\mE'\left(\varphi_{RA}^{\otimes mn}\right)\big)\\
&\leq \max_{\varphi\in\mf(RA)}\min_{\mN'\in B_\epsilon^{\varphi^{\otimes nm}}\left(\mN^{\otimes nm}\right)}D_{\max}\big(\mN'\left(\varphi_{RA}^{\otimes mn}\right)\|(\mE^\varphi)^{\otimes n}\left(\varphi_{RA}^{\otimes mn}\right)\big)\\
&\leq \max_{\varphi\in\mf(RA)}D_{\max}\big(\mM^\varphi_{nm}\left(\varphi_{RA}^{\otimes mn}\right)\|(\mE^\varphi)^{\otimes n}\left(\varphi_{RA}^{\otimes mn}\right)\big)\\
&\leq nr_m-\log(1-\delta_{mn})\\
&=nE_{\mf}^{(m)}(\mN)+n\epsilon-\log(1-\delta_{mn}).
\end{align}
Hence, 
\be\label{gg}
\frac{1}{nm}\underline{LR}_{\mf}^{\epsilon,mn}(\mN)\leq \frac{1}{m}E_{\mf}^{(m)}(\mN)+\frac{\epsilon}{m}-\frac{1}{nm}\log(1-\delta_{mn}).
\ee
Now, similar to the arguments given in~\cite{BP2010} in the state domain, also here we have for any $m\in\mbb{N}$
\be
\limsup_{n\to\infty}\frac{1}{nm}\underline{LR}_{\mf}^{\epsilon,mn}(\mN)=\limsup_{n\to\infty}\frac{1}{n}\underline{LR}_{\mf}^{\epsilon,n}(\mN^{\otimes n}).
\ee
Hence, taking on both sides of~\eqref{gg} the limit $n\to\infty$ followed by the limit $\epsilon\to 0$ gives
\be
\lim_{\epsilon\to 0}\limsup_{n\to\infty}\frac{1}{n}\underline{LR}_{\mf}^{\epsilon,n}(\mN^{\otimes n})\leq \frac{1}{m}E_{\mf}^{(m)}(\mN)\;.
\ee
Since the above equation holds for all $m\in\mbb{N}$, this completes the proof. The proof of the equality for $D_{\mf}^{(\infty)}=LR^{(\infty)}_{\mf}$ follows the exact same lines with $\mD(RA)$ replacing $\mf(RA)$ everywhere.
\end{proof}

\section{Proof of Theorem~\ref{stein}.}

Recall the two types of errors:
\begin{enumerate}
\item The observer guesses that the channel belongs to $\mf(A^n\to B^n)$, 
      while the channel really is $\mN^{\otimes n}_{A\to B}$. This occurs with probability
\be
\alpha^{(n)}(\mN,P_n,\varphi_{RA})\eqdef\tr\left[\mN^{\otimes n}_{A\to B}\left(\varphi_{RA}^{\otimes n}\right)(I-P_n)\right].
\ee
\item The observer guesses that the channel is $\mN^{\otimes n}_{A\to B}$, while the channel 
      really is some $\mM_n\in\mf(A^n\to B^n)$. This occurs with probability
\be
\beta^{(n)}(P_n,\mM_n,\varphi_{RA})\eqdef\tr\left[\mM_{n}\left(\varphi_{RA}^{\otimes n}\right)P_n\right],
\ee
and the worst case for a given $\varphi_{RA}\in\mf(RA)$ is
\be
\beta_{\mf}^{(n)}(P_n,\varphi_{RA})\eqdef\max_{\mM_n\in\mf(A^n\to B^n)}\tr\left[\mM_{n}\left(\varphi_{RA}^{\otimes n}\right)P_n\right].
\ee
\end{enumerate}

We further define
\be
\beta_{\mf,\epsilon}^{(n)}\left(\mN,\varphi_{RA}\right)\eqdef\min\left\{\beta_{\mf}^{(n)}(P_n,\varphi_{RA})\;:\;\alpha^{(n)}(\mN,P_n,\varphi_{RA})\leq\epsilon\;\;;\;\;0\leq P_n\leq I_{R^nB^n}\right\}.
\ee

\begin{theorem*}
Let $\mf$ be a closed convex resource theory admitting the tensor product structure, with the set of free states containing a full rank state. Then, for all $\epsilon\in(0,1)$ and all $\varphi\in\mf(RA)$
\be
\lim_{n\to\infty}-\frac{\log\beta_{\mf,\epsilon}^{(n)}\left(\mN,\varphi_{RA}\right)}{n}
=\lim_{n\to\infty}\min_{\mM\in\mf(A^n\to B^n)}\frac{D\left(\mN_{A\to B}^{\otimes n}(\varphi_{RA}^{\otimes n})\big\|\mM_{A^n\to B^n}(\varphi_{RA}^{\otimes n})\right)}{n} . \nonumber
\ee
\end{theorem*}

\begin{proof}
Fix $\varphi\in\mf(RA)$ and define 
\be
\mM_n(\varphi)\eqdef\left\{\mM_{A^n\to B^n}\left(\varphi_{RA}^{\otimes n}\right)\;:\;\mM\in\mf(A^n\to B^n)\right\}.
\ee
We show that the set $\mM_n(\varphi)$ satisfies the 5 properties of~\cite{BP2010}:
\begin{enumerate}
\item $\mM_n(\varphi)\subset\mD(R^nB^n)$ is closed and convex. This holds trivially since $\mf(A^n\to B^n)$ is closed and convex. 
\item $\mM_n(\varphi)$ contains a state $\sigma_{RB}^{\otimes n}$ with $\sigma\in\mD(RA)$ being full rank.
Indeed, note that by taking $\mM_n=\Omega_{A\to B}^{\otimes n}$ with $\Omega\in\mf({A\to B})$ we get that $\big(\Omega_{A\to B}(\varphi_{RA})\big)^{\otimes n}\in\mM_n(\varphi)$.
Further, taking $\Omega_{A\to B}$ to be the constant channel outputting the fixed full rank state $\omega_B\in\mf(B)$ we get that $\Omega_{A\to B}(\varphi_{RA})=\varphi_R\otimes\omega_B$ is a full rank state.
\item For every $\gamma\in\mM_{n+1}(\varphi)$ then $\tr_k(\gamma)\in\mM_n(\varphi)$ for any $k=1,...,n+1$. Indeed, suppose $\gamma\in\mM_{n+1}(\varphi)$. Then, 
\be
\gamma_{R^{n+1}B^{n+1}}=\mM_{A^{n+1}\to B^{n+1}}\left(\varphi_{RA}^{\otimes n}\otimes\varphi_{RA}\right).
\ee
Now, by tracing out the last subsystem $RB$ we get that
\be
\omega_{R^{n}B^{n}}\eqdef\tr_{RB}\left[\gamma_{R^{n+1}B^{n+1}}\right]=\tr_B\circ\mM_{A^{n+1}\to B^{n+1}}\left(\varphi_{RA}^{\otimes n}\otimes\varphi_{A}\right).
\ee
Define $\Omega\in\cptp(A^n\to B^n)$ as
\be
\Omega_{A^n\to B^n}(X_{A^n})\eqdef\tr_B\circ\mM_{A^{n+1}\to B^{n+1}}\left(X_{A^n}\otimes\varphi_{A}\right)\quad\forall\;X_{A^n}\in\mB(\mH_{A^n})\;.
\ee
Now, since $\mf$ is a QRT admitting the tensor product structure, and since $\varphi_A$ is free, it follows that $\Omega\in\mf(A^n\to B^n)$ (i.e. $\Omega$ is free). Hence,
\be
\omega_{R^nB^n}=\Omega_{A^n\to B^n}(\varphi_{RA}^{\otimes n})\in\mM_n(\varphi)\;.
\ee
The same conclusion holds if we traced out from $\gamma_{R^{n+1}B^{n+1}}$ any of the $n+1$ $RB$ systems.
\item If $\gamma\in\mM_n(\varphi)$ and $\eta\in\mM_m(\varphi)$ then $\gamma\otimes\eta\in\mM_{n+m}(\varphi)$. Indeed, write 
$\gamma_{R^nB^n}=\mM_{A^n\to B^n}\left(\varphi_{RA}^{\otimes n}\right)$ and 
$\eta_{R^mB^m}=\Omega_{A^m\to B^m}\left(\varphi_{RA}^{\otimes m}\right)$. Then, denote by $\Delta_{A^{n+m}\to B^{n+m}}\eqdef\mM_{A^n\to B^n}\otimes \Omega_{A^m\to B^m}\in\mf(A^{n+m}\to B^{n+m})$ and note that 
\be
\gamma_{R^nB^n}\otimes\eta_{R^mB^m}=\Delta_{A^{n+m}\to B^{n+m}}\left(\varphi_{RA}^{\otimes nm}\right)\in\mM_{n+m}(\varphi)\;.
\ee
\item If $\gamma\in\mM_n(\varphi)$ then $P^{\pi}\gamma P^{\pi^{-1}}\in\mM_n(\varphi)$ for every permutation $\pi\in S_n$. Recall that we assume that $\mf$ has the property that if $\mM_n\in\mf(A^n\to B^n)$ then also 
\be
\Pi^{\pi^{-1}}_{B^n\to B^n}\circ\mM_n\circ\Pi^{\pi}_{A^n\to A^n}\in\mf(A^n\to B^n),
\ee
where
\be
\Pi^{\pi}_{A^n\to A^n}(X_{A^n})=P^\pi X_{A^n}P^{\pi^{-1}},
\ee
with $\{P^\pi_{A^n}\}$ a representation of the permutation group in $\mH_{A}^{\otimes n}$. 
Then, for any permutation
$\pi\in S_n$
\ba
P^{\pi}_{R^nB^n}\Big(\mM_{A^n\to B^n}\left(\varphi_{RA}^{\otimes n}\right)\Big)P^{\pi^{-1}}_{R^nB^n}
&=\Pi^{\pi}_{R^nB^n}\circ\left(\id_{R^n}\otimes\mM_{A^n\to B^n}\right)\left(\varphi_{RA}^{\otimes n}\right)\\
&=
\Pi^{\pi}_{R^nB^n}\circ\left(\id_{R^n}\otimes\mM_{A^n\to B^n}\right)\circ\Pi^{\pi^{-1}}_{R^nA^n}\left(\varphi_{RA}^{\otimes n}\right)\\
&=\left(\id_{R^n}\otimes\Pi^{\pi}_{B^n}\circ\mM_{A^n\to B^n}\circ\Pi^{\pi^{-1}}_{A^n}\right)\left(\varphi_{RA}^{\otimes n}\right)
\in\mM_n(\varphi).
\ea
\end{enumerate}
Since the set $\mM_n(\varphi)$ satisfies all the 5 properties of~\cite{BP2010}, the main result of~\cite{BP2010}, which includes both the direct part and strong converse, can be applied to $\mM_n(\varphi)$. In particular, it follows that for any $\epsilon\in(0,1)$
\be
\lim_{n\to\infty}-\frac{\log\beta_{\mf,\epsilon}^{(n)}\left(\mN,\varphi_{RA}\right)}{n}= \lim_{m\to\infty}\frac{1}{m}\min_{\mM\in\mf(A^m\to B^m)}D\left(\mN_{A\to B}^{\otimes m}(\varphi_{RA}^{\otimes m})\big\|\mM(\varphi_{RA}^{\otimes m})\right)\;.
\ee
This concludes the proof.
\end{proof}

\section{Lower bound on the Chernoff bound}
Suppose Alice is given with $t_0$ probability the channel $\mN^{\otimes n}_{A\to B}$ and with $t_1$ probability one of the channels in $\mf(A^n\to B^n)$. Alice's goal is to determine if she is holding in her lab $\mN^{\otimes n}_{A\to B}$ or one of the channels in $\mf(A^n\to B^n)$. The probability of error is therefore given by 
\be
P_{error}^{(n)}(\varphi)=\max_{\mM_n\in\mf(A^n\to B_n)}\frac{1}{2}\left(1-\left\|t_0\mN^{\otimes n}_{A\to B}(\varphi_{R^nA^n})-t_1\mM_n(\varphi_{R^nA^n})\right\|_1\right).
\ee
We had to maximize the error over all possible channels in $\mf$ to get the worst case scenario. She will therefore choose $\varphi$ to minimize the above quantity. That is,
\be
P_{error}^{(n)}\equiv\min_{\varphi\in\mf(R^nA^n)}P_e(\varphi)=\frac{1}{2}\left(1-\max_{\varphi\in\mf(R^nA^n)}\min_{\mM_n\in\mf(A^n\to B^n)}\left\|t_0\mN^{\otimes n}_{A\to B}(\varphi_{R^nA^n})-t_1\mM_n(\varphi_{R^nA^n})\right\|_1\right).
\ee

In~\cite{Audenaert-2007} it was shown that for any two positive operators $A$ and 
$B$ and $\alpha\in(0,1)$ we have
\be
  \tr[A^\alpha B^{1-\alpha}]\geq\frac{1}{2}\tr[A+B-|A-B|]\;.
\ee
Hence, for any $0\leq \alpha\leq 1$,
\be
\frac{1}{2}\left\|t_0\mN^{\otimes n}_{A\to B}(\varphi_{R^nA^n})-t_1\mM_n(\varphi_{R^nA^n})\right\|_1\geq \frac{1}{2}-t_0^\alpha t_1^{1-\alpha}\tr\left[\left(\mN^{\otimes n}_{A\to B}(\varphi_{R^nA^n})\right)^\alpha\left(\mM_n(\varphi_{R^nA^n})\right)^{1-\alpha}\right],
\ee
so that
\be
P_{error}^{(n)}\leq t_0^\alpha t_1^{1-\alpha}\max_{\varphi\in\mf(R^nA^n)}\min_{\mM_n\in\mf(A^n\to B^n)}\tr\left[\left(\mN^{\otimes n}_{A\to B}(\varphi_{R^nA^n})\right)^\alpha\left(\mM_n(\varphi_{R^nA^n})\right)^{1-\alpha}\right].
\ee
We therefore conclude that 
\be
  \liminf_{n\to\infty}-\frac{1}{n}\log P_{error}^{(n)}
    \geq \max_{\alpha\in[0,1]}(1-\alpha)G_{\mf,\alpha}^{\infty}(\mN),
\ee
where
\be
  G_{\mf,\alpha}^{\infty}(\mN) 
    \eqdef \liminf_{n\to\infty}\frac{1}{n}\min_{\varphi\in\mf(R^nA^n)}
                                 \max_{\mM_n\in\mf(A^n\to B^n)}
           D_{\alpha}\left(\mN^{\otimes n}_{A\to B}(\varphi_{R^nA^n})\big\|\mM_n(\varphi_{R^nA^n})\right),
\ee
where $D_\alpha$ is the Petz quantum Renyi divergence.

\end{document}